%% file: paper.tex
\documentclass[9pt]{article}
\usepackage{algorithmic}
\usepackage{algorithm}
\usepackage{graphicx,url,cite,amsmath,amsthm,amssymb,amsfonts,mathrsfs, amsthm}

\usepackage{microtype}%if unwanted, comment out or use option "draft"
\usepackage{xargs}
\usepackage{fancyref}
\usepackage{graphicx}

\usepackage{subfig}
\usepackage{subfloat}
\usepackage{float}
\usepackage{sidecap}
\usepackage{lineno}

\input{macros}

\newtheorem{theorem}{Theorem}[section]
\newtheorem{lemma}[theorem]{Lemma}

\begin{document}

\title{Minimizing Uncertainty through Sensor Placement with Angle Constraints}

\author{
Ioana O. Bercea \thanks{Department of Computer Science, University of Maryland, USA} 
\and Volkan Isler \thanks{Department of Computer Science and Engineering, University of Minnesota, USA}
\and Samir Khuller \thanks{Department of Computer Science, University of Maryland, USA} 
}

\maketitle
\begin{abstract}
We study the problem of sensor placement in environments in which localization is a necessity, such as ad-hoc wireless sensor networks that allow the placement of a few anchors that know their location or sensor arrays that are tracking a target. In most of these situations, the quality of localization depends on the relative angle between the target and the pair of sensors observing it. In this paper, we consider placing a small number of sensors which ensure good angular $\alpha$-coverage: given $\alpha$ in $[0,\pi/2]$, for each target location $t$, there must be at least two sensors $s_1$ and $s_2$ such that the $\angle(s_1 t s_2)$ is in the interval $[\alpha, \pi-\alpha]$. One of the main difficulties encountered in such problems is that since the constraints depend on at least two sensors, building a solution must account for the inherent dependency between selected sensors, a feature that generic $\SC$ techniques do not account for. We introduce a general framework that guarantees an angular coverage that is arbitrarily close to $\alpha$ for any $\alpha <= \pi/3$ and apply it to a variety of problems to get bi-criteria approximations. When the angular coverage is required to be at least a constant fraction of $\alpha$, we obtain results that are strictly better than what standard geometric $\SC$ methods give. When the angular coverage is required to be at least $(1-1/\delta)\cdot\alpha$, we obtain a $\bigo{\log \delta}$- approximation for sensor placement with $\alpha$-coverage on the plane. In the presence of additional distance or visibility constraints, the framework gives a $\bigo{\log\delta\cdot\log\OPTsize}$-approximation, where $\OPTsize$ is the size of the optimal solution. We also use our framework to give a $\bigo{\log \delta}$-approximation that ensures $(1-1/\delta)\cdot \alpha$-coverage and covers every target within distance $3R$. 
\end{abstract}

\newpage

\section{Introduction}
% no \IEEEPARstart
\input{introduction}

\section{Algorithmic Framework}
\input{alg_framework}

\section{Other bi-criteria approximations for $\Pror$}
\input{other}

\bibliographystyle{plain}
\bibliography{paper}

\appendix
\makeatletter
\edef\thetheorem{\expandafter\noexpand\thesection\@thmcountersep\@thmcounter{theorem}}
\makeatother
\input{app}

\end{document}

%% file: macros.tex
\newcommand{\SC}{\ensuremath{\textsc{{Set Cover}}}}

\newcommand{\pairS}{\textsc{Pairwise Selection}}
\newcommand{\MinRep}{\textsc{Min Rep}}
\newcommand{\db}{\text{double-sector}}
\newcommand{\dw}{\text{double-wedge}}
\newcommand{\aec}{(\alpha-\epsilon)}

\newcommand{\aeec}{(\alpha-2\epsilon)}
\newcommand{\an}{((1-1/\delta)\cdot \alpha)}

\newcommand{\Pro}{\alpha \ensuremath{\mathsf{-Ang}}}
\newcommand{\Pror}{(\alpha,R)\ensuremath{\mathsf{-AngDist}}}
\newcommand{\Prol}{\alpha \ensuremath{\mathsf{-ArtAng}}}
\newcommand{\bigo}[1]{\mathcal{O}(#1)}

\newcommand{\OPT}{S_{\textsf{OPT}}}
\newcommand{\OPTsize}{k_{\textsf{OPT}}}
\newcommand{\ft}{\ensuremath{\textsc{Fault tolerant }k\textsc{-suppliers}}}

%% file: introduction.tex
Localization is an important necessity in many mobile computing applications. In ad-hoc wireless sensor networks, it centers around the ability of nodes to self localize using little to no absolute spatial information. When a large number of sensors are deployed, it becomes impractical to equip all of them with the capability  of localizing themselves with respect to a global system (such as through GPS). On the other hand, while GPS can be used for localization in outdoor settings, localization in indoor environments remains a challenge~\cite{sheng2015collocation}. In parallel, as ``smart" home, warehouse and factory automation applications gain traction, providing location services in such settings is becoming more important. When mobility is considered, the problem becomes that of tracking a moving target through a sensor network in which a set of sensors must combine measurements in order to detect the location of the target.

From this perspective, a commonly used technique for localization is to employ a small number anchors (or beacons) that know their location and are capable of transmitting it to the other nodes seeking to localize themselves~\cite{wang2005self}. Alternatively, sensors such as cameras or microphone arrays placed in the environment can collect measurements which can then be used to estimate the locations of objects of interest~\cite{zhang2008maximum,kamthe2009scopes}. In both approaches, some of the most popular measurements used are Euclidean distance and angle between pairs of nodes (bearing), such as is the case with Received Signal Strength Indicator (RSSI), time-to-arrival (ToA) or Angle-of-Arrival (AoA). In this context, each target seeking to localize itself has access to Euclidean distances and/or angular measurements relative to the sensors that are in its vicinity. When exact distances or bearings from two sensors to a target are known, localization can be easily performed through the process of triangulation. In practice, however, the inherent sensor measurements are noisy. To account for this, several models of uncertainty have been proposed, such as assuming probability distributions on the uncertainty in prediction over a set of locations to be measured~\cite{cressie2015statistics} or computing Cramer-Rao bounds as lower bounds for the accuracy of calibration~\cite{moses2003self,patwari2005locating,savvides2003error, 1516106}. 

From a geometric perspective, the target/sensor geometry plays a significant role in the quality of location estimates. In this context, another common benchmark for localization performance is the Geometric Dilution of Precision (GDOP) that investigates how the geometry between the sensors and the target nodes amplifies measurement errors and affects the localization error. Savvides et al.~\cite{savvides2003error, 1516106} observe that the error is largest when the angle $\theta$ between two sensors and the target node is either very small or close to $\pi$. The analysis of Kelly~\cite{kelly2003precision} further shows that, when triangulation is used, this angle contributes to the GDOP at a fundamental level. When distance measurements are used in triangulation, the GDOP is proportional to $1/|\sin \theta|$. When angular measurements are used, the GDOP is proportional to $d_1\cdot d_2/|\sin \theta|$, where $d_1$ and $d_2$ are the distances from the anchors to the node. In general, distance information comes from connectivity of the communication graph and depends on the sensing range of the sensor. As such, it can be modeled as a disk or annulus. Bearing information is subject to additive error and can be modeled as a cone. The above measurements become constraints on the possible location of the target, each restricting the set of feasible locations. Intuitively, the quality of localization depends on both the area and the shape of the intersection of the feasible regions defined by the measurements. When the angle $\theta$ is close to $0$ or $\pi$, the intersection becomes unconstrained and the error unbounded. In particular, when the sensors and the target are collinear, localization is impossible. In essence, overall uncertainty is minimized when the sensors are well separated angularly about the target. 
%\begin{figure}
%  \centering
%  \caption{When bearing information is used to determine a target's location, the measurement corresponds to a cone centered at the true bearing from the sensor to the target. Given sensors $s_1,s_2$ and target $t$, let $\theta = \angle{s_1 t s_2}$, $d_1 = d(s_1,t)$ and $d_2 = d(s_2,t)$. When the sensors are too far from the target, the uncertainty becomes unbounded. If $\theta$ is small, the $y$-direction becomes unconstrained. If $\theta$ is large, the $x$-direction becomes unconstrained. }
%  \includegraphics[height=1.70in]{figs/uncertainty.pdf}
%  \label{fig:unc}
%\end{figure}

Inspired by these observations, our paper focuses on the geometry of sensor deployment and asks the question of where should the sensors be placed such that we control the inherent uncertainty in measurement at the GDOP level. Specifically, given a set of candidate sensor locations and a set of possible target locations, what is the minimum number and placement of sensors so as to ensure that the uncertainty in estimating the target's location is below a given threshold for all possible target locations? The question then becomes one of coverage and, to this end, we define an angular constraint which we call $\alpha$-\textit{coverage}: given a parameter $\alpha \in [0,\pi/2]$, each target at position $x$ must be assigned two sensors (at positions $s_1$ and $s_2$) such that the angle $\theta = \angle s_1xs_2$ is in the range $[\alpha, \pi-\alpha]$ (i.e. neither too small nor too big). Bounding the angle in this fashion allows us to upper bound its influence on the GDOP. We then frame the problem of sensor placement as a bicriteria optimization problem:  given a set of possible sensor and target locations, we wish to select the smallest number of sensors that provide $\alpha$-coverage for all target locations. This formulation captures the scenario in which we have discretized the space into finitely many locations. We study this problem for a number of settings in conjunction with other classical constraints such as sensing range and line-of-sight visibility. We address these variants from a theoretical perspective and present a general algorithmic framework that specifically addresses the angular constraint and iteratively obtains better angular guarantees at the expense of larger solution sizes. 

\subsection{Our model.} Formally, we consider the 2D model in which the set of candidate sensor locations is a discrete set $X\subseteq \mathbb{R}^2$ and the set of target locations is a discrete set $T \subseteq \mathbb{R}^2$. Let $m = |X|$ be the number of possible sensor locations and $n=|T|$ be the number of target locations.  The underlying distance function will be the Euclidean $\ell_2$ metric. We consider (unordered) pairs of the form $(s,s')$ where $s \neq s'$, $s \in S$, $s' \in S'$ and $S,S' \subseteq X$ are sets of sensors. We denote the set of such pairs as $S \times S'$. Formally, $S \times S'  = \{ (s,s') | s \neq s', s\in S, s' \in S'\}$. We say that the pair $(s,s')$ \textit{$\alpha$-covers} the target $t$ if the angle $\angle{sts'} \in [\alpha,\pi - \alpha]$. We also say that the set of pairs $S \times S'$ $\alpha$-covers $t$ if there exists a pair $(s,s') \in S \times S', s \neq s'$,  that $\alpha$-covers $t$. When $S' = S$, we simply say that the set $S$ $\alpha$-covers $t$. We say that a pair or a set of pairs $\alpha$-covers a set $T$ of targets when it $\alpha$-covers each element of $T$. In addition, we say that a pair or a set of pairs $(\alpha,R)$-covers a set $T$ of targets \textit{within distance $R$} if, for at least one of the pairs that $\alpha$-covers a target, the distance from both sensors to the target is $\leq R$. Notice that the parameter $\alpha$ defines a certain range: the higher the value of $\alpha$, the smaller the range of possible values that $\angle{sts'}$ can take. 

The \textbf{Minimum Sensor Placement with $\alpha$-coverage} ($\Pro$) problem asks for computing the smallest set $S$ that $\alpha$-covers $T$. When the sensors have \textit{finite sensing range} $R$, the \textbf{Minimum Sensor Placement with $(\alpha,R)$-coverage} ($\Pror$) asks for the smallest set $S$ that $\alpha$-covers $T$ within distance $R$. We also consider the problem that combines $\alpha$-coverage with visibility constraints, as introduced by Efrat et al~\cite{JM}. Given two regions $Q \subseteq P$, the goal is to place a small set of sensors in $P$ that $\alpha$-cover and guard targets in $Q$. In order to obtain a discrete set of sensors, the authors in \cite{JM} impose an arbitrary grid $\Gamma$ on $P$ and consider potential sensor locations situated at the vertices of $\Gamma$. We note however, that such a step is not necessary in our case, since $X$ is already a discrete set. Given a sensor $s \in X$ and a target $t\in T$, we say that $s$ \textit{sees} $t$ if the segment connecting the two does not cross the boundary of $P$ (i.e. $t$ is within line-of-sight of $s$). Two sensors $s_1,s_2$ \textit{$\alpha$-guard} a target $t$ if they both see the target and $(s_1,s_2)$ $\alpha$-covers $t$. We say that a set $S \subseteq X$ $\alpha$-guards $T$ if, for each target $q \in T$, there exist sensors $s_1, s_2 \in X$ such that $(s_1,s_2)$ $\alpha$-guards $q$. We note that this problem is closely related to the \textsc{Art Gallery} problem which asks for a set of sensors in $P$ such that every point in $Q$ is visible from at least one sensor. The \textbf{Art Gallery with $\alpha$-coverage} ($\Prol$) problem therefore asks for the set $S\subseteq X$ of smallest cardinality that $\alpha$-guards $T$.

\subsection{Related work.} When it comes to sensor coverage problems, extensive work has been done although surprisingly few results discuss $\alpha$-coverage. Notable exceptions are the work of Efrat, Har-Peled and Mitchell~\cite{JM}, Tekdas and Isler~\cite{tekdas10placement} and Isler, Khanna, Spletzer and Taylor~\cite{isler05cviu}. As mentioned before, Efrat et al.~\cite{JM} introduce the $\Prol$ problem in which two sensors are required to $\alpha$-guard a target. They present a $\bigo{\log \OPTsize}$-approximation algorithm that guarantees $\alpha/2$-coverage, where $\OPTsize$ is the smallest set of sensors that satisfy the visibility and $\alpha$-coverage constraints. Their algorithm runs in time $\bigo{n\OPTsize^4\log^2 n \log m}$. In contrast, we present a framework that achieves $(1-1/\delta)\cdot \alpha$-coverage. Tekdas and Isler~\cite{tekdas10placement} formalize the angle constraint by requiring that the uncertainty $d_1 \cdot d_2/|\sin\theta|$ from before be smaller than a certain threshold $U$. When the targets are contained in some subset of the plane and the sensors can be placed anywhere (continuous case), they present a $3$-approximation with maximum uncertainty $\leq 5.5U$. We note that $\Pror$ is a generalization of the above problem in the sense in which an algorithm for $\Pror$ can be used in approximately solving the former. When the sensor locations can be chosen from a uniform grid or if they are unconstrained, we can obtain improved constant-factor approximations in the manner of Tekdas et al.~\cite{tekdas10placement}. In particular for their problem definition, we can bypass their bicriteria approximation by using a square grid and placing at most $24 \cdot \OPTsize$ sensors on its vertices while never violating the uncertainty threshold $U$. A similar approach can be used for $\Pro$ and $\Pror$. Finally, Isler et al~\cite{isler05cviu} consider the case in which the sensor locations are already given and one must compute an assignment of sensors to targets that minimizes the total sum of errors. In addition, they require that each sensor be used in tracking only one target. The version relevant to our problem is when the error is defined as $1/\sin\theta$. In the case in which the sensors are equally spaced on a circle, they present a $1.42$-approximation that also applies to minimizing the maximum error.

\begin{table}
\begin{tabular}{|r|l|}
  \hline
  Coverage Problem & Results\\
  \hline \hline
  $\Pro$ & $\bigo{\log \delta}$-approximation with $(1-1/\delta)\cdot
   \alpha$-coverage\\
  \hline 
  $\Pror$ & 
   within distance $R$:\\
   & 
    $\bigo{\log \delta \cdot \log \OPTsize}$-approximation with $(1-1/\delta)\cdot
       \alpha$-coverage \\ \cline{2-2}
   & within distance $3R$: \\
   &  $\bigo{\log \delta}$-approximation with $(1-1/\delta)\cdot
        \alpha$-coverage\\ \cline{2-2}
   & for $\alpha=0$:\\
   & optimal set of sensors that cover everything within distance $(1+\sqrt{3})\cdot R$\\
   & Euclidean $\ensuremath{\textsc{Fault Tolerant }k\textsc{-suppliers}}$:\\
   & Known $3$-approximation~\cite{Khuller}, new $(1+\sqrt{3})$-approximation\\
        
  \hline
  $\Prol$ & Known: \\
  & $\bigo{\log \OPTsize}$-approximation with $\alpha/2$-coverage\cite{JM} \\ \cline{2-2}
  & New: \\
  
 & $\bigo{\log \delta \cdot \log \OPTsize}$-approximation with $(1-1/\delta)\cdot
         \alpha$-coverage \\
 & $\bigo{\log \delta \cdot \log h \cdot \log (\OPTsize \log h)}$-approximation for polygons with $h$ holes       \\ 
  \hline  
\end{tabular}
\caption{Summary of our results. Depending on each problem formulation, $\OPTsize$ denotes the size of the optimal set of sensors. The results hold for any $\delta>1$ and $\alpha \leq \pi/3$.}
\label{tab}
\end{table}

\subsection{Our contributions and methods.} Our contributions are twofold. For the case of $\alpha\leq \pi/3$, we provide a general bi-criteria framework that approximates the angular coverage to arbitrary precision while guaranteeing a good approximation in the size of the solution. Specifically, for any $\delta>1$, we propose an iterative method that guarantees $(1-1/\delta)\cdot \alpha$-coverage and contributes a factor of $\log \delta$ to the approximation factor of the solution size (Section ~\ref{co}). We exemplify the use of this framework in the context of three coverage problems: one in which we just consider angular constraints ($\Pro$), one in which we have additional distance constraints ($\Pror$) and one in which we have visibility constraints ($\Prol$). A summary of our results can be found in Table ~\ref{tab}.

In addition, we present further approximations for the special case in which we have angular and distance constraints (Section ~\ref{ex}). We relax the distance constraints from $R$ to $3R$ and reduce the approximation factor of the solution size from $\bigo{\log \delta \cdot \log \OPTsize}$ to $\bigo{\log \delta}$, while keeping the angular coverage at $ (1-1/\delta)\cdot \alpha$. As shall be explained in Section~\ref{co}, our main framework for dealing with angular constraints requires us to solve the \textsc{Hitting Set} problem for several geometric objects. Instrumental in this approach are the results of Haussler and Welzl~\cite{HausslerW86} and  Br\"{o}nnimann and Goodrich~\cite{BG} which are based on the existence of good $\epsilon$-nets for a given set system. We defer the intricacies of $\epsilon$-nets to later in the paper but mention now that, given an algorithm that computes in polynomial time an $\epsilon$-net of size $\bigo{(1/\epsilon)\cdot g(1/\epsilon)}$, the algorithm of Br\"{o}nnimann and Goodrich~\cite{BG} returns a hitting set of size $\bigo{\tau \cdot g(\tau)}$, where $\tau$ is the size of the optimal hitting set and $g$ is a monotonically increasing sublinear function. In general, our framework uses $\epsilon$-net constructions to directly obtain good hitting sets for our problem. In the special case of the relaxed distance constraint, however, we employ a two step method which we believe is of independent interest and may have other applications. First, we construct an $\epsilon$-net of size $\bigo{\frac{R_I}{R} \cdot \frac{1}{\epsilon}}$, where $R_I$ is the diameter of the space of all possible sensor locations. We then apply the shifting technique of Hochbaum and Maass~\cite{Hochbaum:1985:ASC:2455.214106} to eliminate the dependency on $\frac{R_I}{R}$ and obtain a constant factor approximation for the corresponding \textsc{Hitting Set} problem in which sensors cover targets within distance at most $3R$.

We also consider the case in which $\alpha=0$ and construct a set of optimal size that covers the targets within distance $(1+\sqrt{3})\cdot R$. This particular case remains relevant since it captures the spirit of fault tolerance by requiring two sensors to be assigned to a target. We achieve our result by showing a $(1+\sqrt{3})$-approximation for the more general Euclidean $\ensuremath{\textsc{Fault Tolerant }k\textsc{-suppliers}}$ problem which improves on the existing $3$-approximation by Khuller et al~\cite{Khuller}.

Before describing our results in more detail, we would like to outline the particular difficulties that the angular constraints give rise to and provide some context for our strategy.  For simplicity, we will consider the $\Pro$ problem, noting that all our observations also apply to $\Pror$ and $\Prol$, since they share the angular coverage constraint. One natural way in which we can consider $\Pro$ is as an instance of $\SC$. Let $\OPTsize$ be the size of the optimal solution for $\Pro$ and $k^*$ the corresponding one for $\SC$. For each pair of sensors $(s,s')$, we can define $S_{(s,s')}$ to be the set of targets $t \in T$ such that $(s,s')$ $\alpha$-covers them. $\SC$ asks for the smallest collection of sets whose union covers $T$.  The generic greedy algorithm for $\SC$ constructs a solution of size at most $k^* \cdot \log n$. Since every newly covered target might require a distinct pair of sensors, this leads to an upper bound of ${\OPTsize \choose 2}$ on $k^*$. This approach therefore yields a $\bigo{\OPTsize\cdot\log n}$- approximation for $\Pro$. By exploiting the underlying geometry of the problem, we can improve the above approximation factor to $\bigo{\OPTsize \log \OPTsize}$. Specifically, one can show that each $S_{(s,s')}$ is induced by the symmetric difference of two circles, as shown in Figure \ref{fig:two}.
\begin{SCfigure}
  \centering
  \caption{Any two distinct sensors $s\neq s'$ uniquely determine, for a given $\alpha \in (0,\pi/2]$, two disks $\mathcal{D}_1=\mathcal{D}(O_1, R_{\alpha})$ and $\mathcal{D}_2=\mathcal{D}(O_2, R_{\alpha})$ of similar radius that have $s$ and $s'$ on their boundary. The set of targets that are $\alpha$-covered by the $(s,s')$ is exactly $(\mathcal{D}_1 \cup \mathcal{D}_2) \setminus (\mathcal{D}_1 \cap \mathcal{D}_2) $. The radius of the disks  is $R_{\alpha} = \frac{d(s,s')}{2} \cdot \frac{1}{\sin\alpha}$.}
  \includegraphics[height=1.75in]{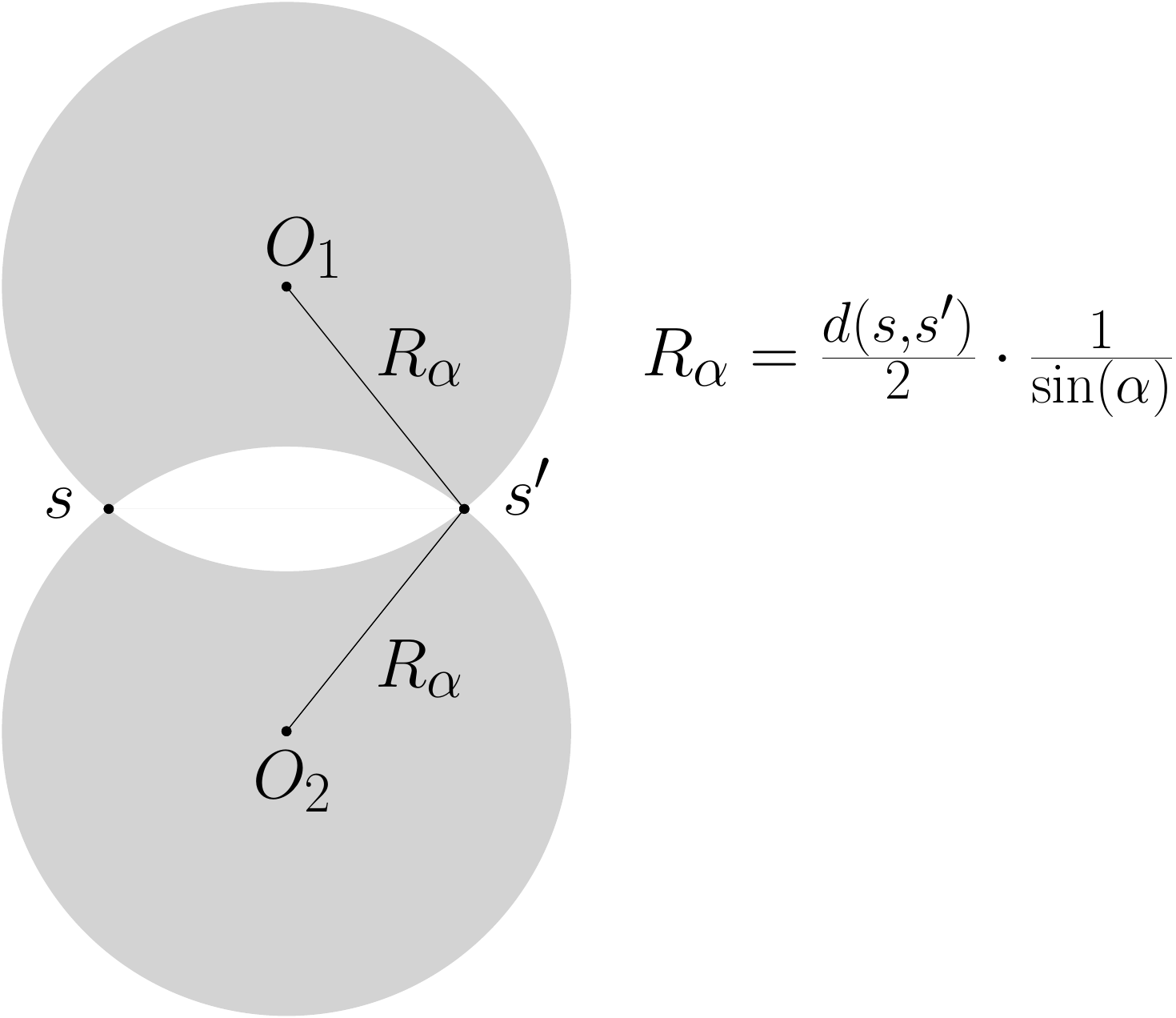}
  \label{fig:two}
\end{SCfigure}
 As such, these objects have constant Vapnik-Chervonenkis (VC) dimension~\cite{Matousekbook}. Given a set system $\mathcal{F}(X, \mathcal{R})$ where $\mathcal{R}$ is a collection of subsets of the universe $X$, the VC dimension of $\mathcal{F}(X, \mathcal{R})$ is the size of the largest subset $Y \subseteq X$ that can be \textit{shattered} in the sense that $|\{ Y \cap S | S \in \mathcal{R} \}| = 2^{|Y|}$. Given a set system with VC dimension $d$, the results of Haussler and Welzl~\cite{HausslerW86} and Br\"{o}nnimann and Goodrich~\cite{BG} imply a method of constructing a set cover of size at most $\bigo{d \log(d \cdot k^*) \cdot k^*}$. When the dimension is constant, this leads to a $\bigo{\log k^*}$-approximation. Unfortunately, this only gets us a $\bigo{\OPTsize \cdot \log \OPTsize}$-approximation guarantee in the total numbers of sensors chosen in the solution.
 
  The persistent $\OPTsize$ factor in the approximation comes from the fact that the $\SC$ framework cannot distinguish between sensors that help cover a lot of targets (in isolation) and sensors that, additionally, can also help cover more targets in conjunction with other sensors. In other words, it does not make use of the global dependency between sensors in order to get a small solution size. Such observations are more in the vain of \textsc{Label Cover} type problems. In fact, as discussed in Appendix ~\ref{pairwise_selection}, when considered in its full generality (i.e. points lie in arbitrary space and coverage is defined arbitrarily), the problem becomes a generalization of \textsc{MinRep} and, as such, incurs a hardness of approximation bound of $2^{ \log^{1-\epsilon} n}$, for any $0<\epsilon<1$ unless NP $\subseteq \text{DTIME}(n^{\text{polylog}(n)})$~\cite{kortsarz2001hardness}. Better approximations must exploit the fact that the contribution of one sensor is measured in terms of all the choices we make in our final solution and must therefore leverage the potential of already chosen sensors when selecting new ones.

In this context, our bi-criteria framework achieves better approximation bounds on the number of sensors used while approximating $\alpha$-coverage to arbitrary precision, for any $\alpha \leq \pi/3$ (Section ~\ref{co}). In particular, when $\delta$ is constant, we get a constant approximation for $\Pro$ and $\bigo{\log \OPTsize}$-approximations for $\Pror$ and $\Prol$, while approximating the angle coverage by a constant. When $\alpha \leq \pi/3$, this directly improves on the geometric $\SC$ approach in the sense in which we obtain a smaller approximation factor in the number of sensors (while approximating the angular converage) and also on the result by Efrat et al.~\cite{JM} for the \textsc{Art Gallery} with $\alpha$-coverage problem, in the sense in which we obtain better angular guarantees while maintaining the same approximation factor in the number of sensors chosen. Our framework is based on iteratively applying a method that, given a set of already chosen sensors $S$ that achieves $\aeec$-coverage (for any $\epsilon < \alpha/2$), selects another set $S'$ that, together with $S$, guarantees a better $\aec$-coverage (Section ~\ref{al}). At its core, our method exploits the collaboration of sensors in $S$ to construct pairs that $\aec$-cover the targets and in which one sensor is in $S$ and the other one in $S'$. The task of constructing $S'$ is then reduced to that of computing hitting sets for set systems induced by geometric objects that have good approximation algorithms. Furthermore, their size will be bounded in terms of $\OPTsize$. Our method significantly generalizes the one used by Efrat et al~\cite{JM} and can be iteratively applied to obtain better and better angular guarantees at the expense of constructing a larger set each time. It is worthwhile to note that the main technical lemma of the framework refers solely to the angle coverage constraint and as such, could be applied to a variety of other problems as long as the other constraints (such as distance or line-of-sight visibility) define a good set system (one with finite VC dimension, for example).

%% file: alg_framework.tex
\label{co}

Our strategy will be to build pairs in a more tractable way: instead of looking at all pairs formed by sensors in $\OPT$, we will look at pairs formed by one sensor in $\OPT$ and another fixed sensor that we have already chosen in our solution (Section ~\ref{al}). The first observation in this line of thought was made by Efrat et al.~\cite{JM}: given an arbitrary set $S \subseteq X$, for every target $t \in T$, there exist sensors $s \in S$ and $s^* \in \OPT$ such that $(s,s^*) $ $\alpha/2$-covers $t$. In other words, $S \times \OPT$ $\alpha/2$-covers $T$. We generalize this observation but require that $\alpha \leq \pi/3$ : 

\begin{lemma} \label{existence}
 Let $\epsilon>0$ be such that $\alpha-\epsilon \leq \pi/3$ and $\epsilon \leq \alpha/2$. Given a set $S$ that $\aeec$- covers $T$, let $T' \subseteq T$ be the set of targets that $S$ does \textbf{not} $\aec$-cover. Then $S \times \OPT$ $\aec$-covers $T'$.
\end{lemma}

In other words, if $S$ does not already $\aec$-cover a target, then the sensors in $S$ can be paired up with sensors in $\OPT$ to $\aec$-cover it. When $\epsilon=\alpha/2$, we start with an arbitrary set $S$ and recover the observation of Efrat et al.~\cite{JM} for $\alpha \leq \pi/3$. We note that, in order to get a better that $1/2$-approximation on the angular coverage, the seed set $S$ cannot be chosen arbitrarily. In fact, our proof crucially uses the power of the pairs in $S$ to $\aeec$-cover the targets, thus further exploiting the collaboration between already chosen sensors.

By fixing some of the sensors to be in $S$ and looking at pairs in $S \times \OPT$,  we can reduce the general problem to a more tractable one of finding a suitable set $S'$ that can achieve what $\OPT$ achieves in this restricted framework(Subsection ~\ref{alcon}). In particular, Lemmas \ref{red_a}, \ref{red_ar} and \ref{red_al} construct the set $S'$ by computing an approximate hitting set. This is where the specific geometry of the other constraints comes into play. Once we have a set $S'$ such that $S \times S'$ $\aec$-covers $T'$, we can add $S'$ to $S$ and obtain a larger set that $\aec$-covers $T$. By iteratively applying this algorithm $\log \delta$ times, we obtain the main results of the paper.

\subsection{Existence of a good solution}
We begin by defining, for each target $t \in T$, sensor $s \in X$ and angle parameter $\beta \in [0,\pi/2]$, the set
$R_t(s, \beta) = \{ s' \in X| (s,s') \text{ }\beta\text{-covers } t\}$. In other words, $R_t(s,\beta)$ represents the set of feasible locations for sensors that, together with $s$, $\beta$-cover $t$ ($\beta$ will be instantiated as $\alpha-\epsilon$). Notice that, from the point of view of $t$, the set $R_t(s, \beta)$ is induced by two wedges around $t$, as seen in Fig. \ref{fig0}.  A \textit{wedge} is defined as the intersection of two non-parallel half spaces in $\mathbb{R}^2$.  Specifically, let $l$ be the line that passes through $s$ and $t$. Let $l_1$ and $l_2$ be the two lines that pass through $t$ and form an angle of $\beta$ with $l$. These two lines describe two opposite wedges of interest: one that corresponds to the half spaces above the lines $l_1$ and $l_2$, and one that corresponds to the half spaces below the lines $l_1$ and $l_2$. The union of these two wedges will be referred to as the \textit{$\dw$} around $t$. The central angle $\theta$ of both the wedges is $\theta = \pi - 2\beta$, and by extension we shall say that the corresponding $\dw$ has a central angle of $\theta = \pi - 2\beta$.

\begin{figure}[t]
\centering
\includegraphics[height=2in]{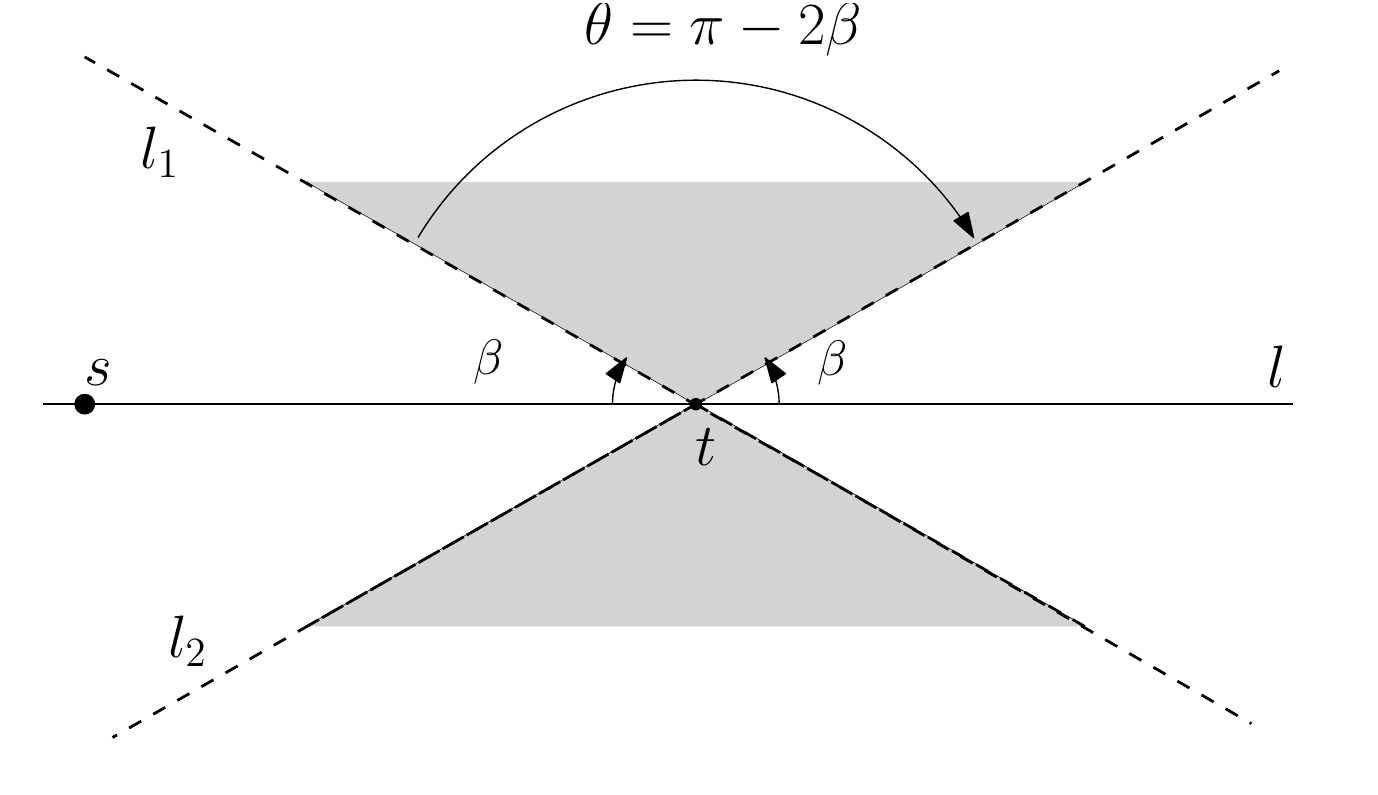}
\caption{The set $R_t(s, \beta)$ is induced by the $\dw$ generated by the lines $l_1$ and $l_2$ and has a central angle $\theta = \pi - 2\beta$. }
\label{fig0}
\end{figure}

In order to prove Lemma \ref{existence}, we essentially show that $\OPT$ must intersect at least one of the $\dw$s generated by $S$ and $T$. First, notice that if a set $S$ already $\aec$-covers a target $t$, then we do not need to worry: $S$ will continue to $\aec-$cover $t$ even when we add $S'$ to $S$. We are therefore concerned with targets in $T'$ that are not already $\aec$-covered by $S$. We note, however, that it is essential that $S$ $\aeec$-covers $T'$. If $S$ does not have this property or is arbitrary, we cannot hope to get past the $\alpha/2$ barrier. 

Fix such a target $t \in T'$ and let $s_1,s_2 \in S$ be any two sensors that $\aeec$-cover $t$ but do not $\aec$-cover it. We will show that there exists $s^* \in \OPT$ such that either $s^* \in R_t(s_1,\alpha-\epsilon)$ or  $s^* \in R_t(s_2,\alpha-\epsilon)$. The candidates will be $s^*_1, s^*_2 \in \OPT$ where $(s^*_1,s^*_2)$ is the optimal pair that $\alpha$-covers $t$. Intuitively, each of the $\dw$s induced by $s_1$ and $s_2$ alone is not big enough to "capture" $s^*_1$ or $s^*_2$. However, if $\angle s_1ts_2$ is in the range $[\alpha-2\epsilon, \pi - (\alpha-2\epsilon)]$, then together, the union $D_t$ of these $\dw$s (which will also be a $\dw$ around $t$) will be sufficiently well spread (i.e. have a large enough central angle ) to guarantee that one of the optimal sensors is contained in it. In other words, at least one of the optimal sensors $s^*_1$ or $s^*_2$ together with either $s_1$ or $s_2$ will $\aec$-cover $t$.

We note, however, that the requirement that $\alpha$ be smaller than $\pi/3$ is relatively tight in this framework, in the sense in which, if $\alpha - \epsilon > \pi/3$, then the central angle of each of the double wedges is too small and we can no longer guarantee that their union $D_t$ forms a bigger double wedge. Furthermore, it is not true that $D_t$ must intersect $\OPT$. 

Consider the $\dw$s $D_1$ and $D_2$ corresponding to $ R_t(s_1,\alpha-\epsilon)$ and $R_t(s_2,\alpha-\epsilon)$, respectively, with central angles $\theta_{D_1} = \theta_{D_2} = \pi-2(\alpha-\epsilon)$. Let $\alpha' = \angle(s_1ts_2)$. 

\begin{lemma}  The union of the two $\dw$s $D_1$ and $D_2$ is a larger $\dw$ $D_t$ centered at $t$ with  central angle $\theta_{D_t} = \pi - 2(\alpha-\epsilon) + \alpha'$.
\end{lemma}

\begin{proof}
We refer the reader to Figure \ref{fig6} for an intuitive explanation. Formally, let $l$ be the line that passes through $s_1$ and $t$ and let $l_1$ and $l_2$ the two lines that define $D_1$. Since $\angle(s_1ts_2) \notin [\alpha-\epsilon, \pi - (\alpha-\epsilon)]$, it follows that $s_2$ is not in $D_1$. Assume without loss of generality that $s_2$ is between the lines $l$ and $l_1$ in the counterclockwise direction. The same proof follows for the other possible locations of $s_2$. 

Now consider $D_2$ and let $l_3$ and $l_4$ be the defining lines through $t$, while $l'$ is the line that passes through $s_2$ and $t$. Notice that $D_1$ and $D_2$ are identical except that $D_2$ is a rotated copy of the $D_1$. In other words, since $\angle(l,l') = \alpha'$, we also have that $\angle(l_1,l_3) = \alpha'$ and $\angle(l_2,l_4) = \alpha'$.
 
Furthermore, since $\alpha' \leq \alpha - \epsilon$ and $\angle( l_1, l_3) \leq \pi - 2(\alpha - \epsilon)$, we have that when $\alpha - \epsilon < \pi /3$, $l_3$ lies in between $l_1$ and $l_2$ and the union of the two $\dw$s $D_1$ and $D_2$ is a continuous $\dw$ $D_t$ determined by $l_1$ and $l_4$.  It has central angle  $\theta_{D_t}  = \theta_1 + \angle(l_2,l_4) = \pi - 2(\alpha-\epsilon) + \alpha'$. 

\end{proof}

\begin{figure}[t]
\centering
\includegraphics[height=2in]{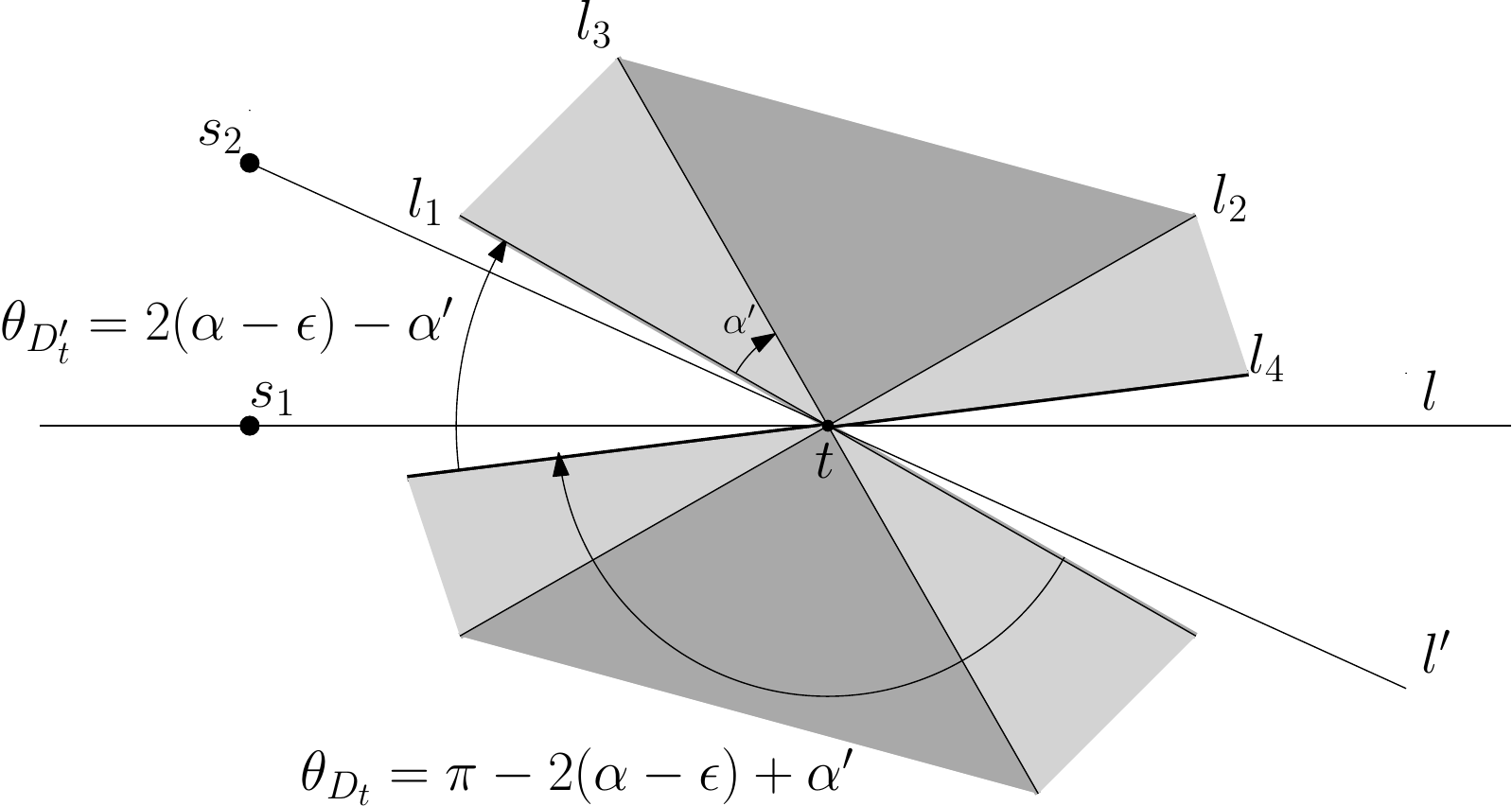}
\caption{Since $\angle(s_1ts_2) = \alpha'$, $D_2$  is a rotation by $\alpha'$ of $D_1$ . Their union is another $\dw$ $D_t$ defined by $l_1$ and $l_4$ with central angle $\theta_{D_t} = \pi - 2(\alpha-\epsilon) + \alpha'$.}
\label{fig6}
\end{figure}

Our goal is to show that one of the two optimal sensors $s^*_1$ and $s^*_2$ must be in $D_t$. The intuition is that by making $D_t$ have a large central angle, we ensure that the complement $D'_t$ of $D_t$ has such a small central angle that it would not be able to contain both $s^*_1$ and $s^*_2$. 

\begin{lemma}  At least one of the two optimal sensors $s^*_1$ and $s^*_2$ assigned to $t$ must be in $D_t$.
\end{lemma}

\begin{proof} Let $D'_t$ be the complement of $D_t$. Notice that $D'_t$ forms another double-wedge defined by $l_1$ and $l_4$ but that it does not actually contain points on these lines.  Moreover, $D'_t$ has a central angle $\theta_{D'_t}=\pi - \theta_D = 2(\alpha-\epsilon) - \alpha'$. 
Since $s_1$ and $s_2$ $(\alpha - 2 \epsilon)$-cover the target, and we are considering the case where $s_2$ is between $l$ and $l_1$, we have that $\alpha' \geq \alpha-2\epsilon$. Hence, we have that $\theta_{D'_t} \leq \alpha$. This implies that $s^*_1$ and $s^*_2$ cannot be both in the same wedge of $D'_t$ without being exactly  situated on the lines $l_1$ and $l_4$( i.e. in $D_t$). The other bad situation would be for them to be in different wedges of $D'$. But then the angle between them would be greater than $\theta_{D_t}$. Since $\alpha' \geq \alpha-2\epsilon$, we get that
 \begin{center} $\theta_{D_t} = \pi - 2(\alpha-\epsilon) + \alpha' \geq \pi - \alpha$,
 \end{center} which would contradict the fact the $\angle(s^*_1ts^*_2) \in [\alpha, \pi-\alpha]$. In other words, at least one of the optimal sensors $s^*_1$ and $s^*_2$ must be in $D_t$.
\end{proof}

\subsection{Construction of the Solution}
\label{alcon}
Once we have determined  that $\OPT$ satisfies the requirements of Lemma \ref{existence}, we will show how to construct a set $S'$ of approximate size that also $\aec$-covers $T'$. As noted before, the proof of Lemma \ref{existence} only talks about angle coverage and is applicable regardless of the other constraints. Depending on the problem at hand, the construction of $S'$ will differ but the general technique is the same. We first illustrate it for $\Pro$ and then mention how to change it for $\Pror$ and $\Prol$. 

In the previous subsection, we described how, for each target $t \in T'$, we can associate a larger $\dw$ $D_t$ that must contain a sensor in $\OPT$. Notice that, by definition, any sensor we pick in $D_t$ will $(\alpha-\epsilon)$-cover $t$ together with a sensor in $S$. We therefore consider the set system induced on $X$ by these $\dw$s $D_t$.  Let $\mathcal{F}(X, \mathcal{R})$ be such a set system, where $\mathcal{R} = \{ D_t \cap X | \hspace{.2cm} t \in T'\}$. In this context, the set $S'$ we are looking for is a hitting set for $\mathcal{F}(X, \mathcal{R})$. In general, a set of sensors $H \subseteq X$ is a hitting set for $\mathcal{F}(X, \mathcal{R})$ if $H \cap D_t \neq \emptyset, \hspace{.2cm} \forall \hspace{.1cm} t \in T'$. The \textsc{Hitting Set} problem asks for a hitting set of minimum cardinality. Let $\tau$ be the size of such an optimal set. Notice that Lemma \ref{existence} shows that $\OPT$ is a hitting set for $\mathcal{F}(X, \mathcal{R})$, so we are guaranteed that $\tau \leq \OPTsize$. Our strategy will therefore be to compute a hitting set of approximate size, and hence obtain a guarantee in terms of $\OPTsize$.

At this point, the geometry of the problems helps even further and we can apply the previously mentioned results of Haussler and Welzl~\cite{HausslerW86} and  Br\"{o}nnimann and Goodrich~\cite{BG}. The technique introduced uses the concept of  an $\epsilon$-net for a given set system. Suppose we are given a set system  $\mathcal{F}(X, \mathcal{R})$. Intuitively, for any $0<\epsilon\leq 1$, an $\epsilon$-net is a set of elements $N \subseteq X$ that intersects all the "heavy" sets of $\mathcal{R}$. In the uniform case, we require $N$ to intersect any set $C\in \mathcal{R}$ with $|C| \geq \epsilon |X|$. More generally, consider a weight function $\mu :X \rightarrow \mathbb{R}_{\geq 0}$ that defines the weight of a subset of $X$ to be the total weight of the points in that subset. An element $A_t \in \mathcal{R}$ is called $\epsilon$-heavy if $\mu(A_t) \geq \epsilon \mu(X)$. An $\epsilon$-net $N$ must then intersect all $\epsilon$-heavy elements of $\mathcal{R}$. In this context, given an algorithm that computes in polynomial time an $\epsilon$-net of size $\bigo{(1/\epsilon)\cdot g(1/\epsilon)}$, the algorithm of Br\"{o}nnimann and Goodrich~\cite{BG} returns a hitting set of size $\bigo{\tau \cdot g(\tau)}$, when $g$ is a monotonically increasing sublinear function. Our strategy will be thus to compute such a small $\epsilon$-net for each specific set system and then use their algorithm as a blackbox to get a $\bigo{g(\tau)}$-approximation for the \textsc{Hitting Set} problem. 

Depending on the underlying geometric objects in the set system, several efficient $\epsilon$-net constructions have been developed.  When the set system has finite VC dimension $d$, Blumer et al~\cite{Blumer} and Koml\'{o}s et al~\cite{Komlos} show that a random sample (under the probability distribution that assigns to $s \in X$ a probability $w(s)/w(X)$ of being sampled) of size $\bigo{(\frac{d}{\epsilon} \log (\frac{1}{\epsilon}))}$ is an $\epsilon$-net with high probability. Better constructions are known for specific set systems: $\epsilon$-nets of size $\bigo{1/\epsilon}$ are known for the case when the underlying objects are halfspaces in $\mathbb{R}^2$ or $\mathbb{R}^3$, pseudo-disks, fat wedges, three-sided axis-parallel rectangles in $\mathbb{R}^2$, and translates of quadrants in $\mathbb{R}^2$ and of fixed convex polytopes in $\mathbb{R}^3$  ~\cite{Laue08,PyrgaR08,Matousek92,KulkarniG10,pach1990some}. 

A particularly simple yet elegant $\epsilon$-net construction was given by Kulkarni and Govindarajan~\cite{KulkarniG10} for the case of $\gamma$-fat wedges. A \textit{$\gamma$-fat wedge} is a wedge having a central angle of at least $\gamma$. When the sets in $\mathcal{R}$ are induced by such $\gamma$-wedges, Kulkarni and Govindarajan ~\cite{KulkarniG10} construct an $\epsilon$-net of size $\bigo{\frac{\pi}{\gamma \epsilon}}$ for arbitrary $\gamma$. When $\gamma \geq \pi/2$, the size of the $\epsilon$-net becomes $\bigo{\frac{1}{\epsilon}}$. This directly applies to the $\Pro$ problem because the wedges in each $\dw$ have a central angle $\theta_D = \pi - 2(\alpha-\epsilon) + \alpha' \geq \pi-\alpha $, since $\alpha' \geq \alpha-2\epsilon$. Therefore, they are $(\pi-\alpha)$-fat. Each $\dw$ can be decomposed into two disjoint wedges, so an $\frac{\epsilon}{2}$-net for $(\pi-\alpha)$-fat wedges is guaranteed to be an $\epsilon$-net for our $\dw$s. Since $\pi - \alpha \geq \pi/2$, we get a hitting set of size $\bigo{\tau}$. Adding this set to $S$, we get Lemma ~\ref{red_a}.

\begin{lemma} \label{red_a}
Let $\epsilon>0$ be such that $\alpha-\epsilon \leq \pi/3$ and $\epsilon \leq \alpha/2$. Let $S \subseteq X$ be a set that $(\alpha-2\epsilon$)- covers $T$. Then we can find a set $S' \subseteq X$ such that $S'$ $(\alpha-\epsilon)$-covers $T$ and $|S'| = |S| + \bigo{1}\cdot \OPTsize$. The running time of the algorithm is $\bigo{\OPTsize \cdot m \log m}$, where $m= |X|$.
\end{lemma}

When we consider the $\Pror$ problem, the set of feasible sensor locations becomes the intersection of a $\dw$ with the circle of radius $R$ centered at the corresponding target. Notice that while the $\dw$ captures the angle requirements (as it did for $\Pro$), the circle represents the additional distance constraints (specific to $\Pror$). Hence, for each target $t$, we define the \textit{$\db$} $C_t = D_t \cap \mathcal{D}(t,R)$. The appropriate set system then becomes $\mathcal{F'}(X, \mathcal{R'})$, where $\mathcal{R'} = \{ C_t \cap X | \hspace{.2cm} t \in T'\}$. Similar to $\dw$s, each $\db$ is composed of two disjoint sectors and hence, an $\frac{\epsilon}{2}$-net for sectors would be an $\epsilon$-net for $\db$s. Since each sector is the intersection of two halfspaces and a circle, each of which induce set systems that have constant VC-dimension, we get that it also has constant VC-dimension ~\cite{Matousekbook}. Therefore, it admits an $\epsilon$-net of size $\bigo{(\frac{d}{\epsilon} \log (\frac{1}{\epsilon}))}$ ~\cite{Blumer},~\cite{Komlos}. This, in turns, gives us a hitting set of size $\bigo{d\tau \cdot \log (\tau)}$ for $\mathcal{F'}(X, \mathcal{R'})$. As before, adding this set to $S$, we get Lemma ~\ref{red_ar}.

\begin{lemma}\label{red_ar}  
Let $\epsilon>0$ and $\alpha$ be as before and $S \subseteq X$ a set that $(\alpha-2\epsilon$)- covers $T$ within distance $R'\geq 0$. Then we can find a set $S' \subseteq X$ such that $S'$ $(\alpha-\epsilon)$-covers $T$ within distance $\max \{R,R'\}$ and $|S'| = |S| + \bigo{\log\OPTsize}\cdot \OPTsize$. The running time of the algorithm is $\bigo{\OPTsize\cdot mn \log m} $.
\end{lemma}

Notice that, in the proof of Lemma ~\ref{existence}, we did not require that the sensors in $S$ cover the targets within the fixed distance $R$. This observation allows us to formulate our results in a more general setting in which $S$ $\aeec$-covers $T$ within some distance $R'$, and will prove useful in Section ~\ref{ex} when we set $R'=3R$. 

In the case of $\Prol$, for each target $t$, the appropriate set system $\mathcal{F''}(X, \mathcal{R''})$ is built by intersecting $D_t$ with the set of sensors $V_t$ in $X$ that guard $t$, called the visibility polygon of $t$. When the underlying polygon $P$ is simply connected, the set system $(X,\{ V_t \cap X | \hspace{.2cm} t \in T'\})$ has constant VC-dimension, as pointed out by Valtr ~\cite{VC23}. It follows that $\mathcal{F''}(X, \mathcal{R''})$ also has finite VC-dimension~\cite{Matousekbook}. As such, it has a hitting set of size $\bigo{\OPTsize \log(\OPTsize)}$. 

%However, this approach would only lead to a running time of $\bigo{\OPTsize\cdot mn \log m} $, due to the application of Br\"{o}nnimann and Goodrich~\cite{BG}. In order to get rid of the linear dependency on $m$, we employ the construction of Efrat et al~\cite{JM} used in finding the set of guards that $\alpha/2$-guards targets. 

\begin{lemma}\label{red_al} 
Let $\epsilon>0$ be such that $\alpha-\epsilon \leq \pi/3$ and $\epsilon \leq \alpha/2$. Let $S \subseteq X \subseteq P$ be a set that $(\alpha-2\epsilon$)- guards $T \subseteq Q \subseteq P$. When $P$ is a simply connected polygon, we can find a set $S' \subseteq X$ such that $S'$ $(\alpha-\epsilon)$-covers $Q$ and $|S'| = |S| + \bigo{\log \OPTsize}\cdot \OPTsize$. The running time of the algorithm is $\bigo{\OPTsize\cdot mn \log m} $.
\end{lemma}

We note that when $P$ has $h$ holes, the VC dimension of the visibility polygon becomes $\bigo{\log h}$ \cite{VC23}. The same approach can therefore be used to construct a set $S'$ of size $\bigo{\log h \cdot \log \OPTsize \cdot \log(\log h \cdot \OPTsize)}$.

\subsection{Iterating to obtain $\an$-coverage}
\label{al}
Given the technical lemmas from before that allow us to refine the angular coverage of a given seed set $S$, we can now develop a more general algorithm that constructs a new set that achieves $\an$-coverage for any $\delta > 1$.  
The idea is to iteratively apply the refinement step (by setting $S = S \cup S'$) $\log \delta$ times, first with $\epsilon = \alpha/2$, then with $\epsilon= \alpha/4$ etc.  At the end of $\log \delta$ iterations, we have that the updated set $S$ $\an$-covers $T$. The running time of the algorithm is $\log \delta$ times the time to find the appropriate hitting set plus the time it takes to find the starting set. This first set (denoted $S_1$) requires special care and depends on the problem at hand. 

Notice that we require $S_1$ to $0$-cover $T$ but one can check that the proof of Lemma $1$ follows in this case even when we do not have two distinct sensors $0$-covering a target. Therefore, in the case of $\Pro$, it suffices to pick $S_1$ to consist of any sensor in $X$ and get the following:

\begin{theorem}\label{app_a} Given $X$, $T$, $\alpha\in[0,\pi/3]$ as above, we can find a set of sensors $S\subseteq X$ such that $S$ $\an$-covers $T$ and $|S| = \bigo{\log \delta}\cdot \OPTsize$.  The running time of the algorithm is $\bigo{\log \delta \cdot \OPTsize \cdot m  \log m}$.
\end{theorem}

When it comes to the $\Pror$ problem, we require that the initial set $S_1$ has the property that each target is within distance $R$ of at least one sensor in $S_1$. Without loss of generality, we can assume that $R=1$ and then our problem becomes an instance of the \textsc{Discrete Unit Disk Cover (DUDC)} problem~\cite{disks}. In \textsc{DUDC},  we are given a set $\mathcal{P}$ of $n$ points and a set of $\mathcal{D}$ of $m$ unit disks in the Euclidean plane. The objective is to select a set of disks $\mathcal{D^*} \subseteq \mathcal{D}$ of minimum cardinality that covers all the points. The problem is a geometric version of $\SC$ and is NP-hard \cite{Garey}. Nevertheless, several constant factor approximations have been developed and all could be used to compute a good approximation while balancing the trade-off between the approximation factor and the running time. For our purposes, we use the $18$-approximation by Das et al~\cite{disks} that has a runtime of $\bigo{n\log n + m\log m +mn}$. We note that better approximations are known, but using them in our framework could increase the total runtime. In each iteration, we increase the size of our set by $\bigo{\log \OPTsize} \cdot \OPTsize$ and since $|S_1| \leq 18\cdot \OPTsize$, we get the following:

\begin{theorem}\label{app_ar}  Given $X$, $T$, $\alpha$ and $R$ as above, we can find a set of sensors $S\subseteq X$ such that $S$ $\an$-covers $T$ within distance $R$ and $|S| = \bigo{\log \delta \cdot \log \OPTsize}\cdot \OPTsize$. The running time of the algorithm is $\bigo{\log \delta \cdot \OPTsize\cdot mn \log m} $.
\end{theorem}

For $\Prol$, we need  to find a set of sensors $S_1 \subseteq X$ that guard $T$. To this extent, we can again employ the fact that the set of visibility polygons has finite VC-dimension \cite{VC23}. Notice that finding a small $S_1$ that guards $T$ corresponds to the hitting set problem for the set system made of sensors and visibility polygons of target locations. We therefore obtain a set of size $\bigo{\log k^*} \cdot k^*$, where $k^*$ is the size of the smallest set of sensors from $X$ that guard $T$. Since $\OPT$ also guards $T$, we have that $k^* \leq \OPTsize$, so we are guaranteed to obtain a solution of size $\bigo{\log \OPTsize} \cdot \OPTsize$. In each iteration, we increase the size of our set by $\bigo{\log \OPTsize} \cdot \OPTsize$ and since $|S_1| = \bigo{\log \OPTsize} \cdot \OPTsize$, we get the following:

\begin{theorem}\label{app_al} Given polygons $Q \subseteq P$,$X$, $T$, and $\alpha \in [0,\pi/3]$ as above, we can find a set of sensors $S \subseteq X$ such that $S$ $\an$-guards $T$ and $|S| = \bigo{\log \delta \cdot \log \OPTsize}\cdot \OPTsize$. The running time of the algorithm is $\bigo{\log \delta \cdot \OPTsize\cdot mn \log m} $.
\end{theorem}

Notice that we can also apply the algorithm only a constant number of times. In particular, we get the following results:
\begin{itemize}
\item for $\Pro$, a constant factor approximation that achieves $(\frac{1}{c}\cdot \alpha)$-coverage for any constant $c>1$ and runs in time $\bigo{\OPTsize \cdot m\log m}$ and,
\item for $\Pror$ and $\Prol$, a $\bigo{\log \OPTsize}$-approximation that runs in time $\bigo{\OPTsize \cdot mn \log m}$ and achieves the same angular coverage as above.

\end{itemize}

We note that, in the case of $\Prol$, we improve on the result of Efrat et al~\cite{JM}, in  that we approximate the $\alpha$-coverage constraint to any constant factor while maintaining the same approximation factor of $\bigo{\log \OPTsize}$. Moreover, our running times are comparable: $\bigo{n \OPTsize^4 \log^2 n \log m}$ in ~\cite{JM} versus $\bigo{\OPTsize \cdot mn\log m}$ for our approximation. We note that their running time comes from the fact they do not directly use the bounded VC dimension of the set system. Instead, they use a previous algorithm designed by Efrat et al~\cite{EfratH02} for approximating the more general \textsc{Art Gallery} problem when the set of targets is restricted to vertices of that grid. When angle constraints are added, they adapt this algorithm to only consider vertices of the grid that also satisfy $\alpha$-coverage. In our scenario in which targets have to be chosen from a discrete set, we do not need to impose a grid and can directly apply the Br\"{o}nnimann and Goodrich algorithm ~\cite{BG}. For the case in which the sensors can be placed anywhere, their algorithm could be employed instead while maintaining the same approximation guarantees.

%% file: other.tex
\label{ex}
The geometric objects at the core of our method are wedges centered at targets whose central angles depend on $\alpha$ and $\epsilon$. These ranges define the set of feasible locations from which we must choose a new set of sensors and are given as input to the afferent hitting set problem. In the case of $\Pror$, the distance constraints require that the sensors we pick be within range $R$ of the target, so our wedges become sectors through intersection with a disk of radius $R$ centered at the target. In this context, we employ the canonical $\epsilon$-net construction of  Blumer et al~\cite{Blumer} and Koml\'{o}s et al~\cite{Komlos} and obtain a $\bigo{\log \OPTsize}$-approximation. 

In an attempt to reduce the approximation factor in this latter case, we consider relaxing the distance constraint and allowing the chosen sensors to be within distance $3R$ of the targets. In other words, we extend the radius of our sectors from $R$ to $3R$. Inspired by the construction of Kulkarni and Govindarajan~\cite{KulkarniG10}, we then propose a deterministic rule for picking sensors and obtain a ``relaxed'' $\epsilon$-net of size $\bigo{\frac{R_I}{R} \cdot \frac{1}{\epsilon}}$, where $R_I$ is the diameter of the largest enclosing ball of all possible sensor locations. We note that this construction is not cyclical: we are not building an $\epsilon$-net for sectors of radius $3R$. In our case, our heavy ranges have the special property that at least $\epsilon \cdot n$ of their points are actually contained inside the sector of radius of $R$. The main difference between our construction and the one in ~\cite{KulkarniG10} is the fact that the objects the latter considers are infinite and, as such, allow for simpler grid-based constructions. In fact, this distinction is indeed the source of the additional $\frac{R_I}{R}$ factor that we incur in our bound. To our knowledge, this is the first $\epsilon$-net construction whose size depends linearly on $\frac{1}{\epsilon}$ and the ratio of the diameter of the input space to the size of the ranges (that is, when size can be appropriately defined). We note that the $\bigo{\frac{1}{\epsilon}}$ construction of Pach and Woeginger~\cite{pach1990some} for translates of convex polygons does implicitly depend on solving the problem for points contained inside a bounded square. It is unclear, however, how to adapt their method for the case in which the ranges are sectors of similar radius but can have arbitrary central angles and orientations. 

In order to overcome this barrier, we look towards the end goal of our construction: that of computing a small hitting set. In this context, we use the shifting technique of Hochbaum and Maass~\cite{Hochbaum:1985:ASC:2455.214106} to first partition our space into cells of bounded width and height and apply our $\epsilon$-net construction to obtain good hitting sets in those restricted spaces. The analysis then yields an overall hitting set of size $\bigo{\OPTsize}$ that achieves the desired $\aec$-coverage and is within distance $3R$ of the targets. The complete argument is rather involved and we defer the exact details to Appendix~\ref{boundedeps}. Formally, we get that:

\begin{theorem}\label{app_arapprox}  Given $X$, $T$, $\alpha>0$ and $R$ as above, we can find a set of sensors $S\subseteq X$ such that $S$ $\an$-covers $T$ within distance at most $3R$ and $|S| = \bigo{\log \delta}\cdot \OPTsize$.  The running time of the algorithm is $\bigo{\OPTsize\cdot mn \log m\log n} $.
\end{theorem}

Another interesting special case of $\Pror$ is the one in which $\alpha=0$,  since it requires us to place two distinct sensors within distance $R$ of each target, in the spirit of fault- tolerance. This is, in general, the \textsc{fault tolerant $k$-suppliers} problem as defined by Khuller et al~\cite{Khuller} that requires us to select $k$ suppliers such that each client has $\delta$ suppliers within an optimal distance $r^*$ of it. Under arbitrary metrics, Khuller et al~\cite{Khuller} develop a $3$-approximation: they select $k$ sensors that are guaranteed to cover the clients within $3\cdot r^*$. Karloff and then Hochbaum and Shmoys~\cite{HochbaumS} show that this factor is tight for the general \textsc{$k$-suppliers} problem (i.e. $\delta =1$), unless P=NP. When the underlying space is $\mathbb{R}^d$ with the $\ell_2$ metric, Nagarajan et al~\cite{euclideanksupplier} improve this factor to $(1+\sqrt{3})$ for Euclidean \textsc{$k$-suppliers}. They crucially use the observation that three clients who are pairwise more than $\sqrt{3}\cdot r^*$ apart from each other can never be covered by the same supplier within distance $r^*$ and reduce the problem of finding $k$ suppliers to that of computing a minimum edge cover. A similar approach can be used in our case, except the structure of the optimal solution corresponds to a simple b-edge cover. In Appendix ~\ref{app1}, we present the details of the analysis and guarantee that $\delta$ suppliers are within distance $(1+\sqrt{3})r^*$ of each client :

\begin{theorem}\label{ksupplierseuc} There exists a polynomial time $(1+\sqrt{3})$-approximation algorithm for the Euclidean $\textsc{fault}$  $\textsc{ tolerant k\textsc{-suppliers}}$ in any dimension for arbitrary $\delta \geq 1$.
\end{theorem}

While the \textsc{$k$-suppliers} problem requires us to minimize the covering radius rather than the size of the set of suppliers, the analysis of the above algorithm actually relies on the existence of $k$ suppliers that cover all the clients within a guess radius $R$ and the algorithm itself never picks more than $k$ sensors. In our case, $k = \OPTsize$, so we will always pick at most $\OPTsize$ sensors. The binary search technique of  Hochbaum and Shmoys~\cite{HochbaumS} can then be used to obtain guarantees with respect to the optimal covering radius $r^*$. The algorithm hence starts with a guess $R$ and returns a set of $k$  sensors that cover everything within distance $ (1+\sqrt{3}) \cdot R$. Since, in the case of $\Pror$, we already know the value of $R$, we automatically get the following result as well:

\begin{theorem}\label{eucl}  Given $X$, $T$, and $R$ as above, we can find a set of sensors $S\subseteq X$ such that $S$ $0$-covers $T$ within distance $R \cdot (1+\sqrt{3})$ and $|S| = \OPTsize$, where $\OPTsize$ is the cardinality of the smallest set of sensors that $0$-covers $T$ within distance $R$. The running time of the algorithm is $\bigo{n^2\log n(m+n \log n)} $.
\end{theorem}

%% file: app.tex
\section{Using $\epsilon$-nets of size $\bigo{ \frac{R_I}{R}\cdot \frac{1}{\epsilon}}$}
\label{boundedeps}

Notice that, in the case of $\Pror$, each range was induced by the intersection of a double wedge and a circle, both centered at the target. As we have previously seen, fat wedges allow for an $\epsilon$-net of size $\bigo{\frac{1}{\epsilon}}$\cite{KulkarniG10}. In essence, this comes from the fact that wedges can extend infinitely in one direction, so there are good rules for constructing an $\epsilon$-net that can guarantee that, eventually, one of its points will intersect an $\epsilon$-heavy wedge. In the case of sectors, the additional requirement that they be bounded by a circle, however, makes achieving a similar result considerably harder. 

In general, one possible technique for building a good $\epsilon$-net is to construct a grid-like structure on top of the input points and pick points (deterministically or randomly) that will be guaranteed to intersect the heavy ranges. Without any guarantees on the sizes of the ranges (for some notion of size), one could split the input points into vertical and horizontal strips that each contain some constant fraction of $\epsilon\cdot n$ points and pick points in each such cell. The heavy ranges will be then guaranteed to intersect at least a constant fraction of these cells, which could then be used to show that they also contain at least one of the chosen points. The direct consequence of allowing both such a horizontal and vertical discretization, however, would be that the number of cells would increase quadratically in $\frac{1}{\epsilon}$, leading to sub optimal bounds on the size of the $\epsilon$-net. Better constructions would have to somehow allow such a discretization to happen in only one direction and present efficient techniques for picking less that $1/\epsilon$ points in each such cell. As an example, we refer the reader to the result of Aronov et al.~\cite{AronovES10} for a powerful construction that builds an $\epsilon$-net of size $\bigo{\frac{1}{\epsilon} \log \log (\frac{1}{\epsilon})}$ for a variety of objects such as axis-parallel rectangles, boxes and $\alpha$-fat triangles.

In this context, the elegant construction of Kulkarni et al~\cite{KulkarniG10} splits the input points in only one such direction (say horizontally) and then picks a constant number of points in each generated strip. The fact that wedges themselves extend indefinitely in the other (complementary) direction is key to showing that they contain one of the chosen points. In order to obtain a similar result and at the same time handle the boundedness of our ranges, we employ the fact that all the sectors have a similar radius $R$. This allows up to further split the points in the perpendicular direction (vertically) but this time, in equally spaced strips of fixed width $R$. Intuitively, one can think of these vertical strips as bounding the horizontal strips in a way that mimics the way our sectors are bounded wedges. This, however, is not enough to ensure that a good rule exists for picking a constant number of points in each cell. In particular, the rule of Kulkarni et al~\cite{KulkarniG10} does not work either. Essentially, this comes from the fact that the intersection of a particular sector with each such cell yields a shape that is rather cumbersome. We deal with this issue by extending the sectors in a way in which each such possible intersection looks roughly like the intersection of an infinite wedge with one of our horizontal strips.  In this context, we bear in mind the fact that our sectors are centered at the target and that any point we pick in our $\epsilon$-net represents a sensor that should respect the angle and distance constraints. To this end, our construction guarantees that we pick sensors that will never violate the angular constraint and will be within $3R$ of the target they are assigned to. 

Finally, we note that, as such, our rule picks a constant number of points in each such cell and yields an $\epsilon$-net of size  $\bigo{ \frac{W}{R}\cdot \frac{1}{\epsilon}}$, where $W$ is basically the width of our input space (i.e. $W = \max \{ x_r-x_l, y_t-y_b\}$, where $x_r$ and $x_l$ (and $y_b$ and $y_l$) are the minimum and maximum over all $x$-coordinates ($y$-coordinates) of points in $X$). In general, $W$ could be as high as $\OPTsize$, so in order to get rid of this dependency, we employ the shitting strategy of Hochbaum and Maass~\cite{Hochbaum:1985:ASC:2455.214106} that allows us to construct a global solution by individually solving the problem on instances of fixed width. This further incurs a constant factor in our final solution size.

We now proceed to formally describe the ingredients. As usual, we will focus on describing an $\epsilon$-net for sectors, noting that this translates into a $2\epsilon$-net for double sectors. 

\begin{theorem}
In the case in which the ranges are induced by sectors of radius $R$ that have central angle in the range $[\pi-\alpha, \pi-\alpha/2]$, for $\alpha \in [0,\pi/2]$, there exists an $\epsilon$-net construction $H$ of size $\bigo{\frac{R_I}{R} \cdot \frac{1}{\epsilon}}$ that guarantees that, for each $\epsilon$-heavy sector, there exists a point in $H$ that intersects the corresponding extended sector of radius $3R$. In this result, $R_I$ represents the radius of the smallest enclosing ball of all the input points.
\end{theorem}

In other words, let $A$ be a generic sector of radius $R$ with central angle $\theta \in [\pi-\alpha, \pi-\alpha/2]$ and let $A'$ be the corresponding sector we obtain by extending the radius to $3R$. If $A$ is $\epsilon$-heavy, we will have that the set $H$ intersects $A'$ non trivially. 

\begin{proof}
Consider an arbitrary system of coordinate axes. We first decompose each sector $A$ into smaller ones that have one side parallel to the coordinate axes, which we call \textit{axis parallel sectors}. Because its central angle is smaller than $\pi$, $A$ will decompose into at most $3$ smaller axis parallel sectors, each with central angle at most $\pi/2$. If we build an $\epsilon$-net for axis parallel sectors, this will turn into a $3\epsilon$-net for the more general type of sectors, so we will restrict out attention to the former. Notice that there are $8$ different types of axis parallel sectors. We will describe the construction for one such type and note that it can be intuitively modified to work for the other types. In particular, we will look at sectors that are formed by the intersection of one horizontal halfspace and another one defined by a line with positive slope, as seen in Fig. ~\ref{fig:fig3}(a). 

For reasons that will become evident later in the proof, however, we must require that all possible axis-parallel sectors either have central angle of at most $\pi/3$ or exactly $\pi/2$. To that end, we consider $2$ additional coordinate axes that are rotated copies of the initial coordinate axes, each by a factor of $\pi/3$, one in clockwise and the other one in counterclockwise direction. It can be easily checked that now, given any sector $A$, there exists a system of coordinate axes that will decompose $A$ into at most $3$ axis parallel sectors that all have central angles that are either $\leq \pi/3$ or exactly $\pi/2$. By constructing separate $\epsilon$-nets for each such coordinate axes, we incur an additional factor of $3$ in our solution size.

We now proceed to describe the rule for picking points. In this part, our construction is similar to the one by Kulkarni et al~\cite{KulkarniG10}, so we briefly summarize the common elements and explain the differences in more detail. The main idea is to divide the input points into $2/\epsilon$ horizontal slices each containing $\epsilon n/2$ points. Each slice is numbered $i$, $1\leq i \leq 2/\epsilon$, from bottom to top in terms of the $y$-coordinate to the highest. For each slice, let $P_i$ be the set of points contained in or above slice $i$ and let $H_i$ denote the convex hull of those points. Kulkarni and Govindarajan~\cite{KulkarniG10} associate with $H_i$ an ordering $H'_i$ of the points on its boundary, in counterclockwise direction starting with the point of highest $y$-coordinate. For every point $p \in H'_i$, they define $N(p) \in H'_i$ to represent the point right after $p$ in $H'_i$, with $N(p)$ being the first element of $H'_i$ in the case when $p$ is the last element of the ordering. They then consider the restriction of $H'_i$ to slice $i$ and define $H''_i$ to be the maximal subsequence of $H'_i$ that consists of points on the boundary of $H_i$ that belong to slice $i$. This set is not empty since it must contain the points of lowest $y$-coordinate in $H_i$, which are in slice $i$. Notice that the points in $H''_i$ go in counterclockwise direction, essentially from left to right. The rule for picking points in the $\epsilon$-net is the following: for each slice $i$, pick the point $p_i$ that is the \textit{last} point in $H''_i$ and its corresponding $N(p_i)$. Notice that the two  points essentially correspond to the rightmost vertices in the convex hull whose segment crosses slice $i$. This leads to an $\epsilon$-net of size $4/\epsilon$. At this point, it is straightforward to see that any $\epsilon$-heavy wedge that fully contains slice $i$ and some other slice above it must either contain $p_i$ or $N(p_i)$ or both. We refer the reader to Fig.~\ref{fig:fig3}(a) for an intuitive explanation. In addition, Fig.~\ref{fig:fig3}(b) shows how the argument fails when we consider sectors instead of wedges.

\begin{figure}
\centering
\subfloat[The argument for wedges.]{\includegraphics[height = 2in]{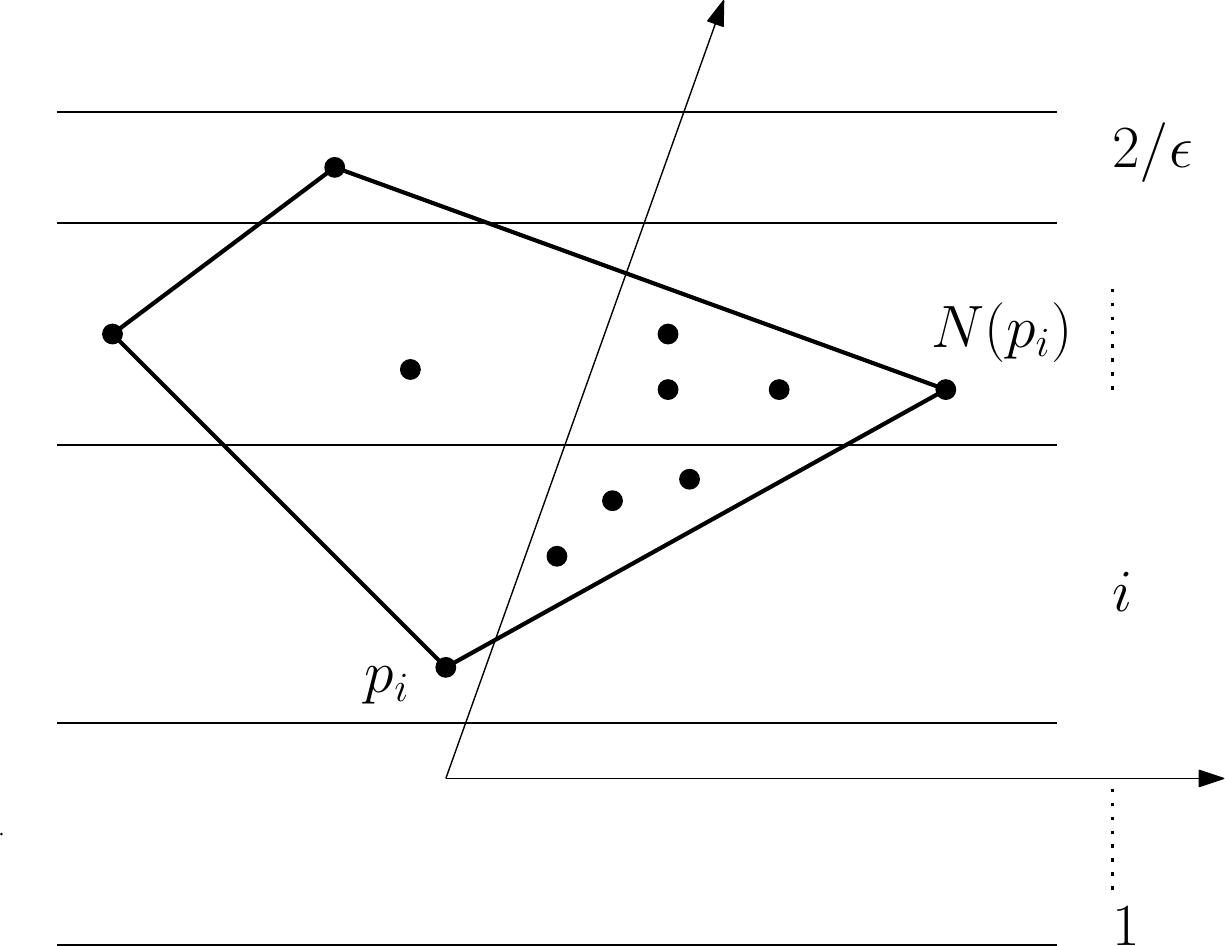}} 
\hspace{.25cm}
\subfloat[The argument for sectors.]{\includegraphics[height = 2in]{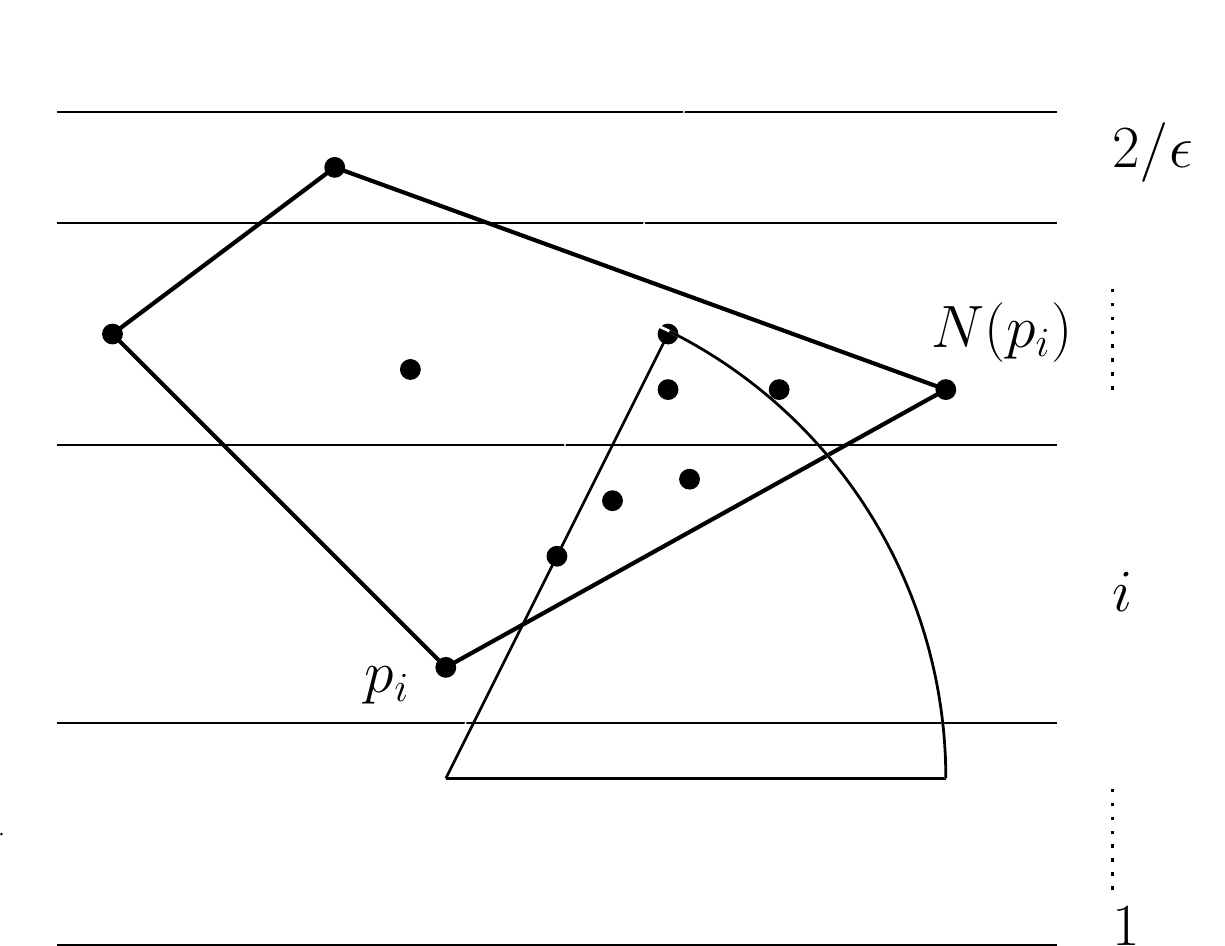}}

\caption{In the case of wedges, the point $N(p_i)$ must be contained in the wedge. However, that is no longer the case for a sector.}
\label{fig:fig3}
\end{figure}

In this context, we modify the above construction in the following ways: first, we will have $4/\epsilon$ horizontal slices each containing $\epsilon n/4$ points each. We will further divide the input space into \textit{vertical strips} of width $R$ each.  The horizontal slices will be numbered $i$,$1 \leq i \leq \lceil 4\epsilon \rceil$. The vertical strips will be numbered $j$, $1\leq j\leq \lceil \frac{x_l - x_r}{R} \rceil $, where $x_l$ and $x_r$ represent the leftmost and rightmost $x$-coordinates of the input points. Each slice $i$ and strip $j$ therefore define the block $B_{ij}$ and let $P_{ij}$ denote the points contained in all the blocks $B_{i'j}$ with $i' \geq i$. In other words, $P_{ij}$ contains all the points in or above slice $i$ restricted to strip $j$. 

Let $H_{ij}$ denote the convex hull of the points in $P_{ij}$, for all $P_{ij} \neq \emptyset$ and, just like before, let $H'_{ij}$  be an ordering of all the points on its boundary, in counterclockwise direction starting with the topmost point. Similarly, let $H''_{ij}$ represent the sequence $H'_{ij}$ restricted to block $B_{ij}$. Notice that $H''_{ij}$ could be empty in the situation in which $B_{ij}$ contains no points. However, we will never deal with such blocks in our proof. Our rule for picking points will still pick $p_{ij}$ and $N(p_{ij})$, where $p_{ij}$ is the last element in $H''_{ij}$ and $N(p_{ij})$ is defined as before. In addition, we also pick the leftmost and rightmost points in each block $B_{ij}$.  The size of the final set $H$ will be $4 \cdot \lceil \frac{4}{\epsilon} \rceil \cdot \lceil \frac{x_l - x_r}{R} \rceil$. 

We will now consider an $\epsilon$-heavy axis parallel sector $A$ and show that $H$ will intersect $A$'s extended sector of radius $3R$. We will focus on the case in which the central angle $\theta$ of $A$ is $\leq \pi/3$. We first begin by noticing that, since $A$ is a sector of radius $R$, it intersects at most $2$ consecutive horizontal strips. In particular, we will have one of two cases, as seen in Fig. \ref{fig:fig5}. 

\begin{figure}
\centering
\subfloat[Case 1.]{\includegraphics[height = 2in]{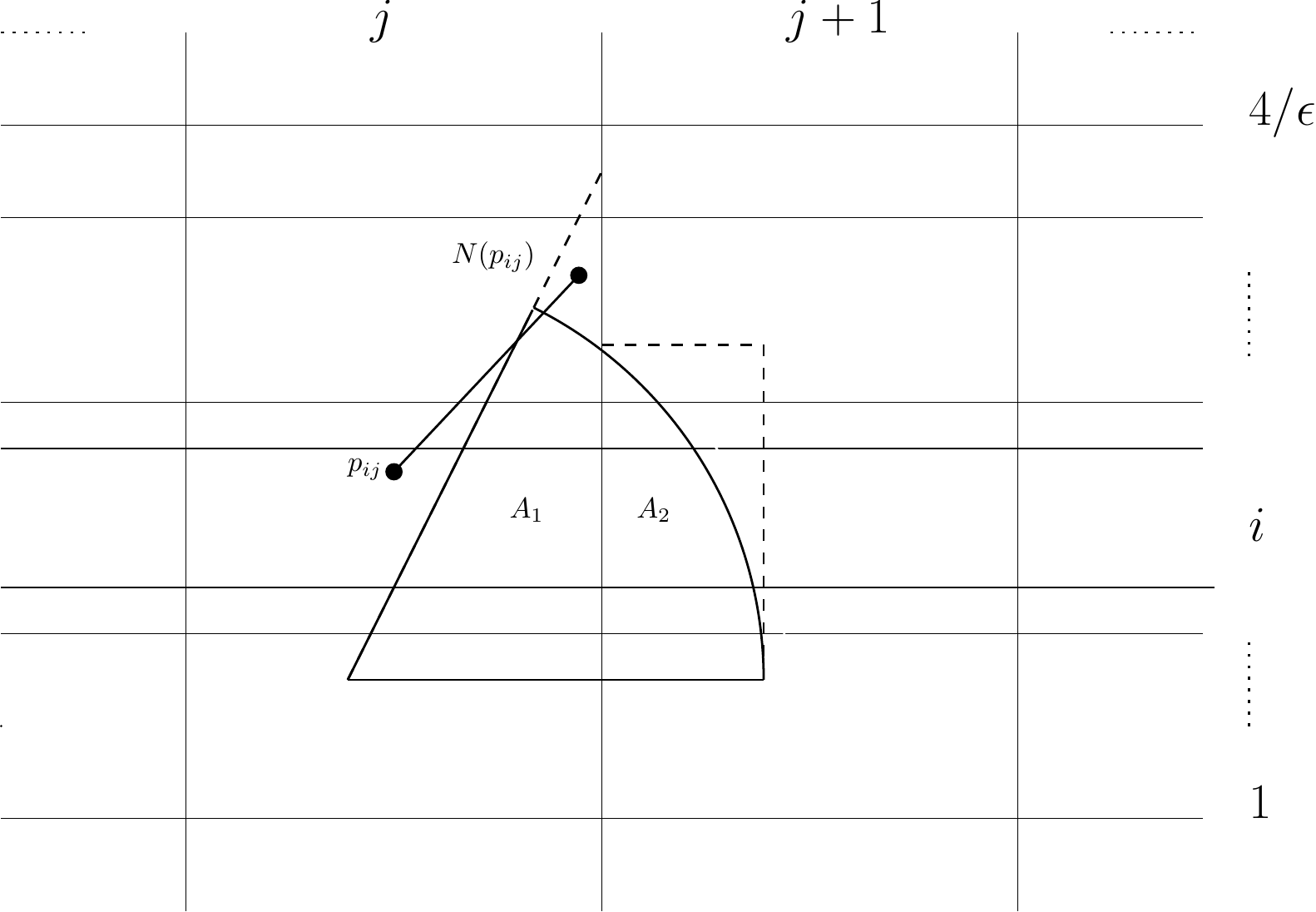}} 
%\hspace{0.03cm}
\subfloat[Case 2.]{\includegraphics[height = 2in]{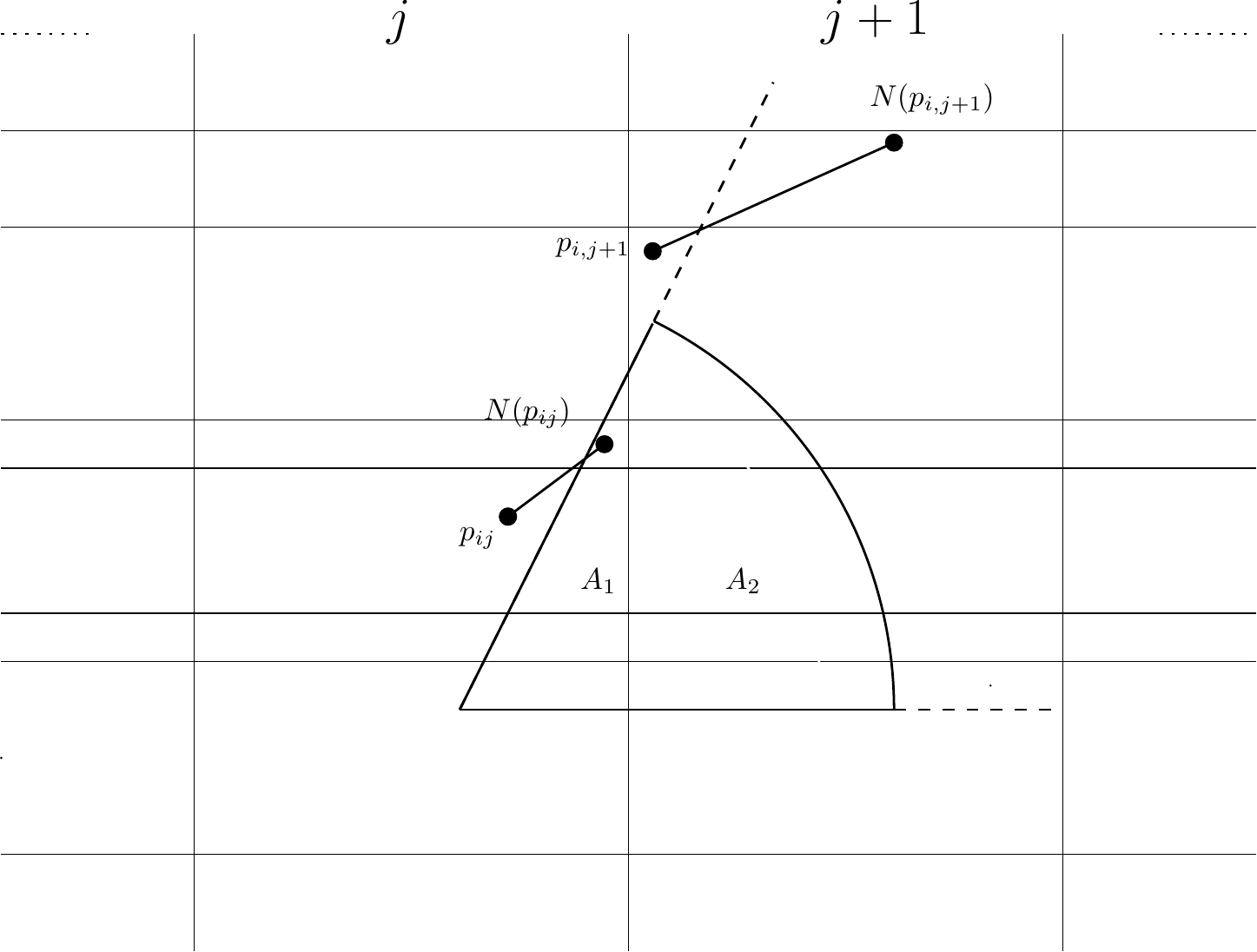}}

\caption{Within each block, $A$ behaves almost like a wedge.}
\label{fig:fig5}
\end{figure}

 \textbf{Case 1.} The vertical line separating strips $j$ and $j+1$ intersects the arc of $A$, as seen in Fig. \ref{fig:fig5}(a). In this case, $A$ decomposes into a left component $A_1$ and a right component $A_2$. One of these two must contain at least $\epsilon n/2$ points. Suppose it is the left component $A_1$. Since each horizontal slice contains exactly $\epsilon n/4$ points,  $A_1$ must intersect at least $2$ horizontal slices (not necessarily consecutive). Let $j$ be the minimum index strip that has a non empty intersection with $A_1$. If $p_{ij}$ is contained in $A_1$, we are in good shape. Otherwise, $N(p_{ij})$ must be contained in the space we get by extending the non horizontal side of $A$ until it meets the closest vertical strip to the right. This is because, otherwise, th convex hull $H_i$ would not cover the points in $A_1$ that are contained in slice $i$ and all the other slices that are above $i$ (of which there is at least one). The farthest possible location for $N(p_{ij})$ is exactly the intersection point with the vertical line. In this situation, the distance from the target is at most $R /\cos \theta)$ and since $\theta \leq \pi/3$, we get that the maximum distance to the target is $2R$.
 
 When the right component $A_2$ has more than $\epsilon n/2$ points, we note that it must also intersect at least $2$ horizontal slices. In this case, we consider the leftmost point in slice $i$. It must be contained in the smallest axis parallel rectangle that encloses $A_2$. The farthest point in this scenario is the upper right corner, which is at most $R \sqrt{1+\sin^2\theta} \leq R\sqrt{2}$  away from the target. 
 
\textbf{Case 2.} The vertical line separating strips $j$ and $j+1$ intersect a side of the sector, as seen in Fig. \ref{fig:fig5}(b). We again decompose into the right and left components $A_1$ and $A_2$ respectively. If $A_1$ contains more than $\epsilon n/2$ points, we are guaranteed that either $p_{ij}$ of $N_{p_{ij}}$ are contained in it, by a similar argument as above. If, on the other hand, $A_2$ contains more than $\epsilon n/2$ points, in the case in which $p_{i,j+1}$ is not contained in it, then $N(p_{i,j+1})$ will be contained in the extended object that we get from intersecting the non-parallel side of $A$ with the closest vertical line to the right. The farthest point in this object is at at most $R (1+1/\cos \theta) \leq 3R$ away from the target.

The other case we need to consider is when $A$ has a central angle of $\pi/2$. In that case, we extend it to the smallest enclosing square (of side length $R$). When the square intersects the vertical line between strips $j$ and $j+1$, separating the square into $A_1$ and $A_2$, we will have that either $A_1$ will contain the rightmost point in box $B_{ij}$ or $A_2$ will contain the leftmost point in $B_{i,j+1}$. The farthest point from the target is the top right corner of the square and is $R\sqrt{2}$ away from the target. 

We have therefore shown that our construction will always contain points that are close to the targets that correspond to the heavy sectors. Additionally, we must note that all of these points are not just a distance at most $3R$ away from the target: they are also contained in $A$'s extended sector of radius $3R$. In other words, we are guaranteed that the initial angle constraint that $A$ imposes is not violated.

Finally, we must note that, in this specific construction, we consider the horizontal width $W_x = x_r-x_l$ because we are dealing with a specific kind of axis parallel sectors given by a specific choice for our coordinate axes. For the other types of axis parallel sectors, the factor will depend on $W_y = y_t - y_b$, where $y_t$ and $y_b$ are the $y$-coordinates of the top and bottom vertices in the input space. We get that, for a fixed coordinate system, the size of the $\epsilon$-net will be at most $\bigo{1} \cdot \lceil \frac{1}{\epsilon} \rceil \cdot \lceil \frac{W}{R} \rceil$, where $W = \max \{W_x,W_y\}$. Combining this with the other constructions that correspond to the rotated systems of coordinate axes, we get that the size of the $\epsilon$-net will be at most  $\bigo{1} \cdot \lceil \frac{1}{\epsilon} \rceil \cdot \lceil \frac{R_I}{R} \rceil$, where $R_I$ is the diameter of the smallest enclosing ball of all the input points.

\end{proof}

\textbf{Remark $1$.} An alternative way of thinking about the above construction is that by relaxing the radius constraint to $3R$, we can include our double sectors in slightly larger geometric objects for which we construct an $\epsilon$-net of size $\bigo{\frac{R_I}{R} \cdot \frac{1}{\epsilon}}$. The way we extend each sector is deterministic and depends on the choice of coordinate axes and the corresponding type of axis-parallel sector. We note, however, that this is not the same as constructing an $\epsilon$-net for the extended sectors of radius $3R$. Each of our objects are subsets of the extended sector of radius $3R$ but their construction relies crucially on the fact that we were extending from sectors of radius $R$. 

\textbf{Remark $2$.} We note that our choice for the threshold of $\pi/3$ for the central angle of an axis parallel sector was rather arbitrary. As we have seen, it was used in bounding the distance from the farthest point of the $\epsilon$-net guaranteed to be in the extended sectors to the target at the center of the sector. Smaller distance guarantees can be obtained by considering smaller thresholds at the expense of requiring more rotated copies of systems of coordinate axes (which in turn will increase the approximation factor in the overall $\epsilon$-net size). 

We will now show how to use this construction in order to build a set $S'$ that is within a constant of the size of the optimal hitting set and guarantees that the points we pick are contained in the extended sectors of radius $3R$. Specifically, we will show the following:
 
\begin{lemma} 
Let $\epsilon>0$ be such that $\alpha-\epsilon \leq \pi/3$ and $\epsilon \leq \alpha/2$. Let $S \subseteq X$ be a set that $(\alpha-2\epsilon$)- covers $T$ within distance $R' \geq 0$. Then we can find a set $S' \subseteq X$ such that $S'$ $(\alpha-\epsilon)$-covers $T$ within distance $\max \{3R,R'\}$ and $|S'| = |S| + \bigo{1}\cdot \OPTsize$.
\end{lemma}

\begin{proof}
As we have seen before, we can construct $S'$ as a hitting set for a new set of geometric objects that we obtained by extending each double sector. Specifically, given the set system $\mathcal{F'}(X, \mathcal{R'})$ containing sectors of radius $R$, we constructed a new set system $\mathcal{F}'_{\text{new}}(X, \mathcal{R}'_{\text{new}})$ for which we gave an $\epsilon$-net construction of size $\bigo{\frac{R_I}{R} \cdot \frac{1}{\epsilon}}$. First notice that the size $\tau^*_{\text{new}}$ of the optimal hitting set for $\mathcal{F}'_{\text{new}}$ is upper bounded by the size $\tau^*$ of the optimal hitting set for $\mathcal{F'}$. A straightforward application of Br\"{o}nnimann and Goodrich~\cite{BG} would therefore get a set $S'$ with $|S'| = \bigo{\frac{R_I}{R} \cdot \tau^*}$.

In order to get a constant approximation factor instead, we will employ the shifting technique of Hochbaum and Maass~\cite{Hochbaum:1985:ASC:2455.214106} as it applies to our case. Before we describe its application, however, we first note that we will construct $S'$ by first restricting ourselves to each of the systems of coordinate axes that we considered in the previous construction. Each sector in $\mathcal{F'}$ corresponds to exactly one coordinate system, specifically the one in which it has the desired property that it decomposes into at most $3$ smaller sectors that have a central angle that is either $\leq \pi/3$ or exactly $\pi/2$. Each coordinate system $i$ therefore corresponds to a different set of ranges $\mathcal{F}'_{\text{new}, i}(X, \mathcal{R}'_{\text{new}, i})$ that are extensions of the sectors of radius $R$ affiliated with that coordinate system. In other words, 
\begin{center}
 $\mathcal{R}'_{\text{new}} = \mathcal{R}'_{\text{new}, 1} \cup \mathcal{R}'_{\text{new}, 2} \cup \mathcal{R}'_{\text{new}, 3}$.
\end{center} 

As such, we will construct $S'$ by actually constructing three separate hitting sets which we will denote as $S'_1$, $S'_2$ and $S'_3$ and then taking their union.  Let $\tau^*_{\text{new},1}$, $\tau^*_{\text{new},2}$ and $\tau^*_{\text{new},3}$ be the size of the optimal hitting set in each of these set systems. It follows that $\tau^*_{\text{new},1} \leq \tau^*_{\text{new}}$ etc. We will now construct hitting sets such that each $|S'_i|$ is within a constant fraction of  $\tau^*_{\text{new},i}$.

We now fix a coordinate system and notice that the appropriate $\epsilon$-net construction from before has size $\bigo{1} \cdot \lceil \frac{1}{\epsilon} \rceil \cdot \lceil \frac{W}{R} \rceil$, where $W = \max \{W_x,W_y\}$. In other words, for a set of points and ranges that are all contained inside a space with $W \leq k \cdot R$, we would get a hitting set of size at most $\bigo{k}$ times  the size of the optimal hitting set for those points. We will denote this local algorithm as $A$.
 
The shifting technique of Hochbaum and Maass~\cite{Hochbaum:1985:ASC:2455.214106} provides a way of analyzing the global approximation factor we would get by applying algorithm $A$ on instances of bounded $W$ and returning the union of all the hitting sets we compute. We will not reproduce the argument entirely in this paper, only mention how it applies to our case. We first need to define the appropriate partition into smaller instances. Suppose we first divide $X$ into vertical strips of width $l \cdot 6R$, where $l \geq 2$ is an arbitrary parameter.  We now need to decide how we will appropriately partition the ranges in $\mathcal{R}'_{\text{new}, i}$. If a range is completely contained in a strip, we assign it to it. Otherwise, assign the range to the strip that contains its center (the target that is the center of the sector). Each such range is contained in a double sector of radius $3R$ which can intersect at most one additional strip.  This is what motivates us to consider strips of size $ l \cdot 6R$ instead of $l \cdot 2R$. For each such strip, we employ an algorithm $A'$ (to be defined later) that will take as input the ranges associated with that strip and the set of sensors contained in the enlarged strip of side length $ (l+1)\cdot 3R$ (i.e. the original strip padded with $3R$ on the left and right).

Now we can see how the global optimal hitting set behaves with respect to the localized optimal hitting sets in each of the strips. We notice that a point in the global optimum can appear at most twice in the induced hitting sets for two consecutive strips. Moreover, when we consider the other possible partitions into strips of width $l \dot 6R$ (obtained by shifting a strip to the right by $6R$), the sets of points that get double counted in each of the partitions are disjoint. This is because, if a point in the optimal global hitting set gets double counted in one partition, it will never get double counted again in any of the shifted $l-1$ partitions since the ranges that contribute to it being double counted are no longer in two distinct strips. Therefore, the analysis of Hochbaum and Maass~\cite{Hochbaum:1985:ASC:2455.214106} applies and we get that, if the algorithm $A'$ solves the hitting set problem for its input strip within a factor of $\alpha_l$, we are guaranteed to obtain an overall approximation guarantee of $\alpha_l \cdot (1+\frac{1}{l})$ (by taking the minimum hitting set over all possible $l$  partitions). The way we are going to determine $\alpha_l$ is by reproducing the above argument and splitting each horizontal strip of width $(l+1)\cdot 6R$ into vertical strips of length $l \cdot 6R$. At this point, we can use algorithm $A$ and give to it a set of points and ranges that are fully contained in a square of side length $(l+1)\cdot 6R$. At this point, $A$ will return a hitting set that is within a factor of $\bigo{l}$ of the optimum. This, in turn, will mean that the approximation factor for $A'$, $\alpha_l$ is also at most $\bigo{l}$. By choosing $l=2$, we get a constant factor of approximation overall. The actual constants depend directly on the $\epsilon$-net constructions used and an improvement in that direction would directly translate in an improvement for the hitting set algorithm.

\end{proof}

\section{($1+\sqrt{3})$-approximation for Euclidean \textsc{Fault tolerant $k$-suppliers}}
\label{app1}
One can think of $\Pror$ as a clustering problem with angle constraints. In this context, a slightly different but relevant problem is the  $\textsc{fault tolerant}-k-\textsc{suppliers}$ problem, as defined by Khuller et al.~\cite{Khuller}. Before we define the $\textsc{fault tolerant}-k-\textsc{suppliers}$ problem, we would like to first talk about the \textsc{$(k,r)$-center} problem as defined by Bar-Iral et al~\cite{BarIlanKP93}, since it captures nicely the bi-criteria nature of our problem. In this problem, we are given a set of $n$ input points and the goal is to choose $k$ points (centers) out of the input points such that every point is within distance $r$ of at least some chosen center. The underlying structure is a graph and the distance is a given function on pairs of vertices (metric or not). There are two ways to attack this problem, each of which gives rise to different optimization problems. On one hand, one could fix the number of centers to be $k$ and then focus on minimizing $r$: this is the well-known \textsc{$k$-center} problem. On the other hand, one could fix the radius $r$ and consider minimizing $k$: this then becomes the \textsc{$r$-domination} problem \cite{BarIlanKP93}. 

The \textsc{$(k,r)$-suppliers} problem is similar to \textsc{$(k,r)$-center}. The main difference is that  the set of centers must be picked from a separate set. Formally, we are given a bipartite graph $G=(U,V,E)$ in which $U$ is the set of suppliers and $V$ is the set of clients, and a distance function $d : U \times V  \rightarrow \mathbb{R}^{\geq 0}$. In the \textsc{$k$-suppliers} problem, the objective is to find a subset $S \subseteq U$ of suppliers of cardinality $k$ such that all the clients in $V$ are within radius $r$ of a center. Specifically, we desire a set $S$ that satisfies
 ${\min_{S \subseteq U, |S| \leq k} \max_{v \in V} \min_{u\ \in S} d(v,u)}$. The analogue \textsc{suppliers $r$-domination} problem requires a set $S$ of minimum size such that ${\max_{v \in V} \min_{u\ \in S} d(v,u) \leq r}$. An interesting variation of the \textsc{$(k,r)$-suppliers} is the \textsc{fault tolerant $(k,r)$-suppliers} problem, in which we require a client to be close to at least $\delta$ of the suppliers, for a given parameter $\delta \leq k$. Specifically we define
 $d^{(\delta)} (v,S) = {\min_{A \subseteq S, |A|= \delta} \max_{u \in A} d(v,u)}$. Then the $\delta\textsc{-neighbors } k\textsc{-suppliers}$ problem (or alternatively, the  $\textsc{fault tolerant}-k-\textsc{suppliers}$ problem with parameter $\delta$) requires us to find a set $S \subseteq U$ of cardinality $k$ such that $S$ satisfies
 ${\min_{S \subseteq U, |S| \leq k} \max_{v \in V} d^{(\delta)}(v, S)}$.Conversely, the $\delta\textsc{-neighbors suppliers } r\textsc{-domination}$ problem requires us to find a set $S \subseteq U$ of minimum cardinality such that ${ \max_{v \in V} d^{(\delta)}(v, S)} \leq r$. Then the  $\delta\textsc{-neighbors suppliers } r\textsc{-domination}$ problem with $\delta=2$ and $r=R$ is exactly $\Pror$ with $\alpha=0$.
 
 Before we begin to describe the $(1+\sqrt{3})$-approximation for $k$-suppliers, we would like to discuss the reason why classical techniques for $k$-center do not work for the case of general $\alpha$. In general, techniques for $k\textsc{-center}$ provide approximation guarantees by first showing a lower bound on the optimal solution size. In particular, they are based on the observation that if two targets are far apart, they must be assigned to different sensors. Specifically, the size of any maximal set of such far apart targets is a good lower bound for the size of the optimal set of sensors. In other words, if we pick one sensor for each such far away target, we are guaranteed to never pick more sensors than the optimal solution would. The algorithm then proceeds to focus on such far apart targets and assigns a sensor to each of them optimally. The rest of the targets must be, by definition, close to one of the far apart targets and get assigned to whatever center is closest to the latter. The cost of focusing on these far apart targets is that the rest of the targets has to now travel farther than what they would have had to travel in the optimal solution. This approach gives the best approximation possible unless P=NP. 
 
 In the fault tolerant case, when the sensors collaborate to cover the targets, a similar approach can be used successfully. Furthermore, in the case of $\Pror$, the observations from before also hold. The problem, however, comes from the fact that we can no longer make the argument that a target can be covered by any sensor at the price of paying for a larger distance. The choice of sensors that we make for the far apart targets might be optimal for those specific targets, but we cannot guarantee that they will help $\alpha$-cover any of the other targets.  When the sensors can be placed anywhere, we can get around this problem by adding only a small constant number of new sensors (per each far apart target) that we can guarantee will cover the rest of the targets, wherever they might be.  In contrast, when the sensors are restricted to only certain locations, we can no longer reason about which of the targets they will be able to cover. In general, each of the other targets might require two additional new sensors to be chosen in our solution. This, in turn, makes it hard to lower bound the size of the optimal solution and provide guarantees on the number of sensors the algorithm would choose. Another way of thinking about this issue is that, in $\Pror$, the relative position of the sensors with respect to the target they are supposed to cover is encoded as a constraint rather than a parameter to be optimized (such as distance is in $k$-\textsc{center}). In this context, an approximation algorithm that relaxes that constraint would have to exploit its inherent geometric structure so that it can cover the rest of the targets in a way that can be translated quantitatively as a function of the size of the optimal solution.

In the $(1+\sqrt{3})$-approximation for $k$-suppliers, the key observation made by Nagarajan et al~\cite{euclideanksupplier} is the following:
\begin{lemma} (Nagarajan et al~\cite{euclideanksupplier}) If three clients have pairwise distances strictly greater than $\sqrt{3} \cdot R$, then they cannot be covered within distance $R$ by the same supplier.
\end{lemma} 

Notice that this is a more advanced observation that the classical one we have mentioned before. In that observation, only two far apart targets were considered at each time (i.e. with distance greater than $2R$) with the guarantee that they cannot share the same vertex. In the Euclidean context, this observation can be extended to three far apart targets as long as we require that the pairwise distance between them is greater than $\sqrt{3}R$. This is where the particular structure of the Euclidean space comes into place, since, in general metric spaces, this observation does not hold true.

The authors hence restrict their attention to a maximal set $P$ of clients that have pairwise distances $> \sqrt{3} \cdot R$. They construct the graph $G$ that has $P$ as the vertex set and an edge between two clients in $P$ if they can be covered within distance $R$ by the same supplier in $U$. Since any supplier can serve at most two clients in $P$, there exists a one-to-one correspondence between suppliers and edges in $G$ and the problem of clustering the clients in $P$ becomes equivalent to finding an minimum size edge cover of $G$. Once they compute a set $S \subseteq U$ of suppliers that cover the clients in $P$ within distance $R$, all the other clients are within distance $\leq \sqrt{3}$ away from $P$ and hence, within distance $\leq (1+\sqrt{3})$ away from $S$.

In the case of $\ft$, the same observation applies and hence, we still have a one-to-one correspondence between edge in $G$ and suppliers in $U$. The structure of the optimal solution, however, is different. It corresponds to a \textsc{simple $b$-edge cover}: a subset of edges such that each vertex $v\in P$ is incident on at least $b_v$ edges, for all $v\in P$. In our case, $b_v = \delta$. It is known that computing a minimum size \textsc{simple $b$-edge cover} problem can be done in polynomial time by computing a maximum size \textsc{$b$-matching} \cite{CombOpt}. The latter can be solved in time $\bigo{n^2\log n(m+n \log n)}$, where $n$ is the number of vertices and $m$ is the number of edges \cite{Anstee1987153}. The analysis that obtains a $(1+\sqrt{3})$-approximation remains the same.

\section{\textsc{Pairwise Selection}}
\label{pairwise_selection}

In this section, we would like to provide a more general problem formulation that captures the inherent difficulties associated with satisfying angular constraints. In particular, we will model the observation that a target is not covered unless both sensors that $\alpha$-cover it are selected. In other words, while the objective function keeps track of the number of sensors we select, the specific constrains to be satisfied depend on selecting pairs of sensors that can then $\alpha$-cover the targets.

In this context, consider the graph $G = (V,E)$ defined on the set of possible sensor locations (i.e. with $V = X$) having an edge between two sensors $s,s' \in X$ if and only if the pair $(s,s')$ $\alpha$-covers some target in $T$. For each target $t \in T$, we look at all the pairs of sensors that $\alpha$-cover it. By definition, this will be a subset $E_t$ of the edges in $E$.  The task of picking the smallest number of sensors that $\alpha$-cover $T$ then becomes equivalent to picking the smallest set of vertices $S \subseteq V$ with the property that, for each target $t \in T$, there exist two sensors $s,s' \in S$ such that $(s,s') \in E_t$. In other words, for each possible target $t$, there exist a pair of sensors in $S$ that $\alpha$-cover it.

 We formalize this as the \textsc{Pairwise Selection} problem:

 \begin{center}
  \noindent\framebox{\begin{minipage}{5.5in}
  \textbf{\textsc{Pairwise Selection}}\\
  \emph{Input }: A graph $G = (V,E)$ and a collection $E_1,E_2,\ldots,E_l$ of $l$ subsets of edges, $E_i \subseteq E$, for all $i \in \overline{1,l}$. \\
  \emph{Output} : A set $S \subseteq V$ of mininum cardinality such that $E[S] \cap E_i \neq \emptyset$, for all $i \in \overline{1,l}$.
  \end{minipage}}
  \end{center}
  \medskip
  
By $E[S]$ we mean the edges induced by the subset $S$, i.e. an edge $e = (u,v) $ is in $E[S]$ if both of its endpoints are in $S$, $u,v \in S$. We now show that the \textsc{Min Rep} problem is a special case of \pairS. In \MinRep, we are given a bipartite graph $G=(A,B,E)$, where $|A|=|B|=n$. Each of the sets $A$ and $B$ are partitioned in $k$ sets of size $q = \frac{n}{k}$ each, $\mathcal{A} = \{A_i | i\in\{1,\ldots,k\}\}$ and $\mathcal{B}=\{B_i | i \in \{1,\ldots, k\}\}$. Given this, we form a the bipartite \textit{supergraph} $H$ in which the vertices represent the sets $A_i$ and $B_i$. Vertices $A_i$ and $B_j$ are connected by a \textit{superedge} if there exist elements $a_i \in A_i$ and $b_j \in B_j $ that are adjacent in $G$. In this situation, we say that the edge $(a_i,b_j)$ \textit{covers} the superedge $(A_i,B_j)$. The $\MinRep$ problem asks for the set $S= A' \cup B'$ of minimum size such that the pairs $(a',b'), a' \in A', b' \in B'$ cover all the superedges of $H$.

Any such instance of $\MinRep$ can be transformed into an instance of $\pairS$ in the following way: for each superedge $(A_i,B_j)$ define the set 
\begin{center}
 $E_{i,j} = \{ (a,b) | a\in A_i, b \in B_j\}$.
\end{center}

In this construction, a min size solution to $\MinRep$ is equivalent to a minimum size solution to \pairS. In other words, if there exists an $\alpha$-approximation algorithm for \pairS, then we would immediately get an $\alpha$-approximation algorithm for \MinRep.  Kortsarz~\cite{kortsarz2001hardness} showed, however, that there is a hardness of $2^{ \log^{1-\epsilon} n}$, for any $0<\epsilon<1$ for $\MinRep$ unless NP $\subseteq \text{DTIME}(n^{\text{polylog}(n)})$, and hence, for $\pairS$ as well. In terms of positive results, there is a $O(n^{1/3} \log^{2/3}n)$-approximation algorithm for $\MinRep$ due to Charikar et al~\cite{charikar2011improved}. We believe it would be an interesting challenge to try to extend this to $\pairS$.

%
% In this context, we say that the edge $(u,v)$ is \textit{satisfied} or, more specifically, the edge $(u,v)$ satisfies the set $E_i$ if $u$ and $v$ are both chosen in the set and $(u,v) \in E_i$. This is because we can associate with each vertex $v$ a boolean variable $x_v$ and then each set $E_i$ can be described as the clause: $\displaystyle{\bigvee_{(u,v)= e \in E_i} }(x_v \wedge x_u)$.

%% file: paper.bbl
\begin{thebibliography}{10}

\bibitem{Anstee1987153}
Richard~P. Anstee.
\newblock A polynomial algorithm for b-matchings: An alternative approach.
\newblock {\em Information Processing Letters}, 24(3):153 -- 157, 1987.

\bibitem{AronovES10}
Boris Aronov, Esther Ezra, and Micha Sharir.
\newblock Small-size {$\epsilon$-Nets for Axis-Parallel Rectangles and Boxes}.
\newblock {\em {SIAM} J. Comput.}, 39(7):3248--3282, 2010.

\bibitem{BarIlanKP93}
Judit Bar{-}Ilan, Guy Kortsarz, and David Peleg.
\newblock How to allocate network centers.
\newblock {\em J. Algorithms}, 15(3):385--415, 1993.

\bibitem{Blumer}
Anselm Blumer, Andrzej Ehrenfeucht, David Haussler, and Manfred~K. Warmuth.
\newblock Learnability and the {Vapnik-Chervonenkis} dimension.
\newblock {\em J. {ACM}}, 36(4), 1989.

\bibitem{BG}
H.~Br\"{o}nnimann and M.T. Goodrich.
\newblock Almost optimal set covers in finite{ VC-dimension}.
\newblock {\em Discrete \& Computational Geometry}, 14(1):463--479, 1995.

\bibitem{charikar2011improved}
Moses Charikar, MohammadTaghi Hajiaghayi, and Howard Karloff.
\newblock Improved approximation algorithms for label cover problems.
\newblock {\em Algorithmica}, 61(1):190--206, 2011.

\bibitem{cressie2015statistics}
Noel Cressie.
\newblock {\em Statistics for spatial data}.
\newblock John Wiley \& Sons, 2015.

\bibitem{disks}
Gautam~K. Das, Robert Fraser, Alejandro L\'{o}pez-Ortiz, and Bradford~G.
  Nickerson.
\newblock On the discrete unit disk cover problem.
\newblock In {\em WALCOM: Algorithms and Computation}, Lecture Notes in
  Computer Science. Springer Berlin Heidelberg, 2011.

\bibitem{JM}
A.~Efrat, S.~Har-Peled, and J.S.B. Mitchell.
\newblock Approximation algorithms for two optimal location problems in sensor
  networks.
\newblock In {\em BroadNets'05}, pages 714--723 Vol. 1, Oct 2005.

\bibitem{EfratH02}
Alon Efrat and Sariel Har{-}Peled.
\newblock Guarding galleries and terrains.
\newblock {\em Inf. Process. Lett.}, 100(6):238--245, 2006.

\bibitem{Garey}
Michael~R. Garey and David~S. Johnson.
\newblock {\em Computers and Intractability: A Guide to the Theory of
  NP-Completeness}.
\newblock W. H. Freeman \& Co., New York,USA, 1979.

\bibitem{HausslerW86}
David Haussler and Emo Welzl.
\newblock Epsilon-nets and simplex range queries.
\newblock In {\em Symposium on Computational Geometry}, pages 61--71, 1986.

\bibitem{Hochbaum:1985:ASC:2455.214106}
Dorit~S. Hochbaum and Wolfgang Maass.
\newblock Approximation schemes for covering and packing problems in image
  processing and vlsi.
\newblock {\em J. ACM}, 32(1):130--136, January 1985.

\bibitem{HochbaumS}
Dorit~S. Hochbaum and David~B. Shmoys.
\newblock A unified approach to approximation algorithms for bottleneck
  problems.
\newblock {\em J. ACM}, 33(3):533--550, May 1986.

\bibitem{isler05cviu}
V.~Isler, S.~Khanna, J.~Spletzer, and C.J. Taylor.
\newblock Target tracking with distributed sensors: The focus of attention
  problem.
\newblock {\em Computer Vision and Image Understanding Journal},
  (1-2):225--247, 2005.

\bibitem{kamthe2009scopes}
Ankur Kamthe, Lun Jiang, Matthew Dudys, and Alberto Cerpa.
\newblock Scopes: Smart cameras object position estimation system.
\newblock In {\em Wireless Sensor Networks}, pages 279--295. Springer, 2009.

\bibitem{kelly2003precision}
Alonzo Kelly.
\newblock Precision dilution in triangulation based mobile robot position
  estimation.
\newblock In {\em Intelligent Autonomous Systems}, volume~8, pages 1046--1053,
  2003.

\bibitem{Khuller}
Samir Khuller, Robert Pless, and Yoram~J. Sussmann.
\newblock Fault tolerant k-center problems.
\newblock {\em Theoretical Computer Science}, 242(1–2):237 -- 245, 2000.

\bibitem{Komlos}
J\'{a}nos Koml\'{o}s, J\'{a}nos Pach, and Gerhard Woeginger.
\newblock Almost tight bounds for $\epsilon$-nets.
\newblock {\em Discrete Comput. Geom.}, 7(2):163--173, March 1992.

\bibitem{kortsarz2001hardness}
Guy Kortsarz.
\newblock On the hardness of approximating spanners.
\newblock {\em Algorithmica}, 30(3):432--450, 2001.

\bibitem{KulkarniG10}
Janardhan Kulkarni and Sathish Govindarajan.
\newblock New $\epsilon$-net constructions.
\newblock In {\em CCCG'10}.

\bibitem{Laue08}
S{\"{o}}ren Laue.
\newblock Geometric set cover and hitting sets for polytopes in $\mathbb{R}^2$.
\newblock In {\em {STACS} 2008, 25th Annual Symposium on Theoretical Aspects of
  Computer Science, Bordeaux, France, February 21-23, 2008, Proceedings}, pages
  479--490, 2008.

\bibitem{Matousek92}
Ji\v{r}{\'i} Matou\v{s}ek.
\newblock Reporting points in halfspaces.
\newblock {\em Comput. Geom.}, 2:169--186, 1992.

\bibitem{Matousekbook}
Ji\v{r}{\'i} Matou\v{s}ek.
\newblock {\em Lectures on Discrete Geometry}.
\newblock Springer-Verlag New York, Inc., Secaucus, NJ, USA, 2002.

\bibitem{moses2003self}
Randolph~L Moses, Dushyanth Krishnamurthy, and Robert~M Patterson.
\newblock A self-localization method for wireless sensor networks.
\newblock {\em EURASIP Journal on Applied Signal Processing}, pages 348--358,
  2003.

\bibitem{euclideanksupplier}
Viswanath Nagarajan, Baruch Schieber, and Hadas Shachnai.
\newblock The {Euclidean} k-supplier problem.
\newblock In {\em Integer Programming and Combinatorial Optimization}, volume
  7801, pages 290--301. Springer Berlin Heidelberg, 2013.

\bibitem{pach1990some}
J{\'a}nos Pach and Gerhard Woeginger.
\newblock Some new bounds for epsilon-nets.
\newblock In {\em Proceedings of the sixth annual symposium on Computational
  geometry}, pages 10--15. ACM, 1990.

\bibitem{patwari2005locating}
Neal Patwari, Joshua~N Ash, Spyros Kyperountas, Alfred~O Hero~III, Randolph~L
  Moses, and Neiyer~S Correal.
\newblock Locating the nodes: cooperative localization in wireless sensor
  networks.
\newblock {\em Signal Processing Magazine, IEEE}, 22(4):54--69, 2005.

\bibitem{PyrgaR08}
Evangelia Pyrga and Saurabh Ray.
\newblock New existence proofs $\epsilon$-nets.
\newblock In {\em Proceedings of the 24th {ACM} Symposium on Computational
  Geometry, College Park, MD, USA, June 9-11, 2008}, pages 199--207, 2008.

\bibitem{1516106}
A.~Savvides, W.L. Garber, R.L. Moses, and M.B. Srivastava.
\newblock An analysis of error inducing parameters in multihop sensor node
  localization.
\newblock {\em Mobile Computing, IEEE Transactions on}, 4(6):567--577, Nov
  2005.

\bibitem{savvides2003error}
Andreas Savvides, Wendy Garber, Sachin Adlakha, Randolph Moses, and Mani~B
  Srivastava.
\newblock On the error characteristics of multihop node localization in ad-hoc
  sensor networks.
\newblock In {\em Information Processing in Sensor Networks}, pages 317--332.
  Springer, 2003.

\bibitem{CombOpt}
Alexander Schrijver.
\newblock {\em Combinatorial optimization. {P}olyhedra and efficiency. {V}ol.
  {A}}.
\newblock Springer-Verlag, Berlin, 2003.

\bibitem{sheng2015collocation}
Kaikai Sheng, Zhicheng Gu, Xueyu Mao, Xiaohua Tian, Weijie Wu, Xiaoying Gan,
  and Xinbing Wang.
\newblock The collocation of measurement points in large open indoor
  environment.
\newblock In {\em Proceedings of the IEEE INFOCOM}, 2015.

\bibitem{tekdas10placement}
O.~Tekdas and V.~Isler.
\newblock Sensor placement for triangulation based localization.
\newblock {\em IEEE Tran. Automation Science and Engineering}, 7(3):681--685,
  2010.

\bibitem{VC23}
Pavel Valtr.
\newblock Guarding galleries where no point sees a small area.
\newblock {\em Israel Journal of Mathematics}, 104(1):1--16, 1998.

\bibitem{wang2005self}
Fu-bao Wang, Long Shi, and Feng-yuan Ren.
\newblock Self-localization systems and algorithms for wireless sensor
  networks.
\newblock {\em Ruan Jian Xue Bao(J. Softw.)}, 16(5):857--868, 2005.

\bibitem{zhang2008maximum}
Cha Zhang, Dinei Flor{\^e}ncio, Demba~E Ba, and Zhengyou Zhang.
\newblock Maximum likelihood sound source localization and beamforming for
  directional microphone arrays in distributed meetings.
\newblock {\em Multimedia, IEEE Transactions on}, 10(3):538--548, 2008.

\end{thebibliography}
